\documentclass{article}
\usepackage{multirow}
\usepackage{fullpage}
\usepackage[bottom]{footmisc}
\usepackage{amsmath}
\usepackage[utf8]{inputenc}
\usepackage[numbers]{natbib}
\usepackage{amsmath,amssymb}
\usepackage{bbm}

\DeclareMathOperator*{\argmax}{arg\,max}

 \usepackage{algorithm}
 
\usepackage{hyperref}
\usepackage{cleveref}

\usepackage{amsthm}
\newtheorem{theorem}{Theorem}[section]
\newtheorem{lemma}[theorem]{Lemma}

\newtheorem{definition}[theorem]{Definition}
\newtheorem{corollary}[theorem]{Corollary}

\usepackage{fullpage}

\usepackage[pdftex]{graphicx}

\usepackage{comment}
\usepackage{mathtools}

\usepackage[shortlabels]{enumitem}

\usepackage[shortlabels]{enumitem}

\newcommand{\ignore}[1]{}

\newcommand{\hide}[1]{}

\usepackage[shortlabels]{enumitem}

\usepackage[shortlabels]{enumitem}

\newcommand{\A}{\mathcal{A}}

\newcommand{\M}{\mathcal{M}}

\usepackage{caption}
\usepackage{xcolor}

\def\Rset{\mathbb{R}}

\DeclareMathOperator*{\conv}{conv}

\DeclareMathOperator{\Reg}{\mathsf{Reg}}

\DeclareMathOperator{\poly}{poly}

\DeclarePairedDelimiter{\norm}{\|}{\|}

\newcommand{\cA}{\mathcal{A}}
\newcommand{\cB}{\mathcal{B}}
\newcommand{\cC}{\mathcal{C}}
\newcommand{\cD}{\mathcal{D}}
\newcommand{\cF}{\mathcal{F}}

\newcommand{\cM}{\mathcal{M}}

\newcommand{\cP}{\mathcal{P}}
\newcommand{\cR}{\mathcal{R}}
\newcommand{\cS}{\mathcal{S}}

\newcommand{\cU}{\mathcal{U}}

\newcommand{\cX}{\mathcal{X}}

\newcommand{\eps}{{\varepsilon}}

\DeclareMathOperator{\BR}{\mathsf{BR}}

\newcommand{\csp}{\varphi}

\newcommand{\esh}[1]{}

\newcommand{\nat}[1]{}
\newcommand\jon[1]{}

\newcommand{\optdist}{\mathcal{D}}

\usepackage{nicematrix}

\newcommand{\values}{\mathcal{F}}

\newcommand{\OPT}{\textsf{OPT}}

\newcommand{\bwabort}{\textsc{BlackwellAbort}}
\newcommand{\abort}{\textsc{Abort}}
\usepackage{algorithmic}

\usepackage[normalem]{ulem}

\title{Learning to Play Against Unknown Opponents}

\author{
Eshwar Ram Arunachaleswaran\thanks{University of Pennsylvania. 
Email: \texttt{eshwarram.arunachaleswaran@gmail.com}. This work was done in part while the author was visiting the Simons Institute for the Theory of Computing }
\and 
Natalie Collina\thanks{University of Pennsylvania.
Email: \texttt{ncollina@seas.upenn.edu} }
\and
Jon Schneider\thanks{Google Research. Email: \texttt{jschnei@google.com}}
}
\begin{document}

\maketitle

\begin{abstract}
    
We consider the problem of a learning agent who has to repeatedly play a general sum game against a strategic opponent who acts to maximize their own payoff by optimally responding against the learner's algorithm. The learning agent knows their own payoff function, but is uncertain about the payoff of their opponent (knowing only that it is drawn from some distribution $\mathcal{D}$). What learning algorithm should the agent run in order to maximize their own total utility, either in expectation or in the worst-case over $\mathcal{D}$? 

When the learning algorithm is constrained to be a no-regret algorithm, we demonstrate how to efficiently construct an optimal learning algorithm (asymptotically achieving the optimal utility) in polynomial time for both the in-expectation and worst-case problems, independent of any other assumptions.
When the learning algorithm is not constrained to no-regret, we show how to construct an $\varepsilon$-optimal learning algorithm (obtaining average utility within $\varepsilon$ of the optimal utility) for both the in-expectation and worst-case problems in time polynomial in the size of the input and $1/\varepsilon$, when either the size of the game or the support of $\mathcal{D}$ is constant. Finally, for the special case of the maximin objective, where the learner wishes to maximize their minimum payoff over all possible optimizer types, we construct a learner algorithm that runs in polynomial time in each step and guarantees convergence to the optimal learner payoff. All of these results make use of recently developed machinery that converts the analysis of learning algorithms to the study of the class of corresponding geometric objects known as menus.

\end{abstract}

\section{Introduction}

How should we design learning algorithms that act on our behalf in repeated strategic interactions (e.g., repeatedly bidding into an auction, playing a repeated game against an opponent, repeatedly negotiating a contract, etc.)?

One approach (arguably, the standard approach in the literature) is to design learning algorithms that satisfy strong worst-case counterfactual guarantees. For example, it is possible to design learning algorithms that guarantee \emph{low (external) regret}: that after $T$ rounds, the gap between the utility achieved by the algorithm and the utility achieved by the best fixed action in hindsight grows sublinearly in $T$. This approach has the strength that it assumes essentially nothing about the behavior of the opponent (that is, it works for any adversarially chosen sequence of opponent actions). But this strength is also, in some ways, a weakness: in most strategic environments, the opponent will be a rational agent who acts not completely adversarially, but to maximize their own payoff. Paradoxically, this can lead to situations where a learning algorithm performs quite poorly despite satisfying these guarantees; for example, \cite{braverman2018selling} show that if a bidder uses certain low regret algorithms to decide how to bid in a repeated auction, the auctioneer can take advantage of this to extract the full welfare of the buyer. Moreover, it is common to have at least some information about our opponent’s incentives, even if we do not understand them completely, and it is wasteful to not take advantage of this information. 

A second approach (the one we take in this paper) is to take inspiration from Bayesian mechanism design and consider the problem of designing a learning algorithm to play against a \emph{prior distribution} of possible opponents. Each opponent in the support of this distribution has their own utility function, is aware of the learning algorithm the learner is using, and (non-myopically) chooses their sequence of actions to maximize their own utility. We would like to design a learning algorithm that maximizes the utility of the learner over the course of this process. In other words, we want to solve a Stackelberg equilibrium problem in \emph{algorithm space}, where the learner leads by committing to a specific learning algorithm, and their opponent follows by playing the sequence of actions that optimally responds to this learning algorithm.

In contrast to the first approach (where efficient algorithms achieving a variety of different regret guarantees have been designed), fairly little is currently understood about this second approach. When the distribution $\cD$ is a singleton distribution (i.e., the learner knows the optimizer's payoff exactly), recent work \cite{zuo2015optimal,collina2023efficient} has resolved the algorithmic question of efficiently finding this optimal commitment. However in the more general setting where the learner possesses some uncertainty about the optimizer's payoff (the setting where one would most naturally want to employ a learning algorithm), this question remains open. 

\subsection{Our Results}

In this paper, we address this question by providing efficient algorithms for this problem in a variety of different contexts. More formally, we consider a setting where two players -- a learner and an optimizer -- repeatedly play a general-sum normal-form game $G$, where the learner has $m$ pure strategies and the optimizer has $n$ pure strategies. In this game, the payoff function $u_L$ of the learner is fixed and known to both parties, but the optimizer’s payoff function $u_O$ is sampled (once, before the game) from a known prior distribution $\mathcal{D}$ with support size $k$. 

At the beginning of the game, the learner commits to playing a specific learning algorithm $\mathcal{A}$, which describes their action at time $t$ as a function of the actions taken by the optimizer from rounds $1$ through $t-1$. The optimizer (of a specific type, sampled from $\mathcal{D}$) best responds to this commitment by playing the sequence of actions which maximizes their own utility. We are now faced with the following computational question: can we efficiently construct a learning algorithm $\mathcal{A}$ that maximizes the utility of the learner? 

In this paper we study this computational question in a handful of different settings. We prove the following main results:

\begin{enumerate}
\item We first study this question under the constraint that our learning-algorithm $\mathcal{A}$ is \emph{no-regret} (obtaining $o(T)$ regret after $T$ rounds). We show that, under this condition, we can efficiently (in polynomial time in $n$, $m$, and $k$) construct a learning algorithm $\mathcal{A}$ that achieves and computes the optimal utility for the learner (Theorem \ref{thm:proof_of_nr_alg}). 

\item We then turn our attention to the problem of constructing the optimal learning algorithm under no additional constraints. In this setting, we provide a method to construct an $\eps$-approximately optimal learning algorithm $\cA$ in time 

$$\left(\frac{mnk}{\eps}\right)^{O(\min(mn, k))}.$$

\noindent
This results in an efficient construction when \emph{either} the size $mn$ of the game or the size $k$ of the support is constant-sized (Theorem \ref{thm:main_result_without_nr}).

\item We then provide some evidence to show that the unconstrained problem above is hard when both $mn$ and $k$ are large. In particular, we prove that it is NP-hard to take as input a $k$-tuple $(u_1, u_2, \dots, u_k)$ of real numbers and decide whether there exists a learning algorithm $\cA$ where the optimizer of type $i$ receives utility at most $u_i$ when best responding to $\cA$ (Theorem \ref{thm:hardness}). This operation is a key subproblem in both of our approaches above, and implies that any efficient construction of an optimal learning algorithm must take a substantially different approach.

\item Finally, complementing this hardness result, we demonstrate an explicit efficient optimal learning algorithm under no additional constraints for the \emph{maximin} objective, where the learner wants to maximize the utility they receive from the worst type of optimizer. Notably, this construction guarantees that the learner receives (nearly) their maximin value $V$ while sidestepping the actual problem of computing $V$ (a problem we conjecture is NP-hard).

\end{enumerate}

Naively, the above computational problems are quite difficult, as they involve optimizing over the space of \emph{learning algorithms}: complex, infinite-dimensional objects defined via sequences of multivariate functions (e.g., the function which maps the history at time $t$ to the action played at time $t$). We handle this by applying a sequence of reductions from this general optimization problem to more tractable optimization problems. 

First, we employ a framework introduced by \cite{arunachaleswaran2024paretooptimal} that allows us to translate utility-theoretic statements about learning algorithms to equivalent statements about convex subsets of $\Rset^{mn}$ called \emph{menus}. At a high level, given any learning algorithm $\cA$, the menu $\cM(\cA) \subset \Rset^{mn}$ contains all distributions over pairs of pure strategies (``correlated strategy profiles'') that the optimizer can asymptotically induce by playing against this learning algorithm. Importantly, these correlated strategy profiles contain all the information necessary to compute quantities such as the utilities of both the learner and optimizer along with various forms of regret incurred by the learner. The menu $\cM(\cA)$ of a learning algorithm $\cA$ therefore forms a sufficient summary of all relevant properties of $\cA$; instead of thinking of an opponent as playing an arbitrary sequence of $T$ actions, we can think of this opponent as selecting a specific correlated strategy profile outcome from the menu $\cM(\cA)$ and inducing it by playing the corresponding sequence of actions. 

This reduces an optimization problem over the set of all algorithms to an optimization problem over a collection of convex sets (possible menus corresponding to a learning algorithm). Although convex sets are much ``nicer'' mathematical objects than learning algorithms, this optimization problem is still quite difficult -- the space of convex sets is still infinite-dimensional, even when they are subsets of a finite-dimensional Euclidean space. Here we reduce the problem further by applying a variant of the revelation principle and considering the resulting matching from the $k$ different optimizer types to the correlated strategy profiles they choose from the menu, which we call a \emph{CSP assignment}. This CSP assignment is $kmn$-dimensional (specifying one correlated strategy profile in $\Rset^{mn}$ for each optimizer type), and so if we can provide an effective description of the set of valid CSP assignments (e.g., via an efficient separation oracle), we can efficiently compute the CSP assignment which optimizes the learner's utility (and with a little more work, the optimal learning algorithm that induces this CSP assignment). 

We are now left with the problem of characterizing the set of valid CSP assignments, which in turn involves understanding the space of valid menus. In the setting where we restrict our algorithm to also be no-regret, we show how to write this set as an explicit $kmn$-dimensional polytope, thus allowing us to compute the optimal algorithm in polynomial time. Our main tool here is a characterization by \cite{arunachaleswaran2024paretooptimal}, showing that the menu of any no-regret algorithm contains the menu of a no-swap-regret\footnote{No-swap-regret is a similar but stronger counterfactual guarantee than no-regret that ensures that the regret is small on every subset of rounds where the learner played a specific action. See Section \ref{sec:model} for a precise definition.} algorithm, that all no-swap-regret algorithms have the same menu, and that all convex sets containing this menu are menus of some algorithm. Intuitively, this means that we can think of any generic no-regret algorithm as directly offering the CSP assignment to their opponent, and ``falling back'' to an arbitrary no-swap-regret algorithm if their opponent ever attempts to deviate. In the more general setting, we no longer have such a strong characterization of unconstrained learning algorithms. Instead, we demonstrate how to construct a separation oracle for the set of CSP assignments by reducing this to the problem of deciding whether a specific set is approachable (in the context of Blackwell Approachability~\cite{blackwell1956analog}) in a specific vector-valued game (with vector-valued payoffs of dimension $\min(mn, k)$). This in turn allows us to use an algorithm of~\cite{mannor2009approachability} that decides the approachability of a set in an arbitrary instance of Blackwell approachability (albeit in time exponential in the dimension of vector-valued payoffs).

Finally, one may ask if it is truly necessary to pay this exponential cost in the dimension $\min(mn, k)$ when designing unconstrained learning algorithms. On one hand, this is a central roadblock to the Blackwell approachability based algorithms above -- \cite{mannor2009approachability} show that the problem of deciding whether a specific set is NP-hard. Nonetheless, we show how to sidestep this obstacle for the specific objective of the maximin utility\footnote{We briefly emphasize here that although we have stated the above results in the context of maximizing the expected utility of the learner when the optimizer is drawn from a known prior distribution, this optimization framework is fairly generic and allows for optimizing many other objectives over the space of learning algorithms (essentially, any objective which can be written as a concave function of the CSP assignment). } of the learner (the worst-case utility against any optimizer in the support of $\cD$). The key idea is to essentially run the same algorithm as before (computing a candidate optimal menu), but to initially be very optimistic about which menus are approachable. As we actually run our learning algorithm, either we will approach this menu without issue, or at some point our approachability algorithm will fail to play a valid action -- at this point, we can restart the algorithm, being slightly more cautious about which menus are approachable. Interestingly, unlike all previous algorithms we have discussed, this algorithm does not allow us to directly compute the true optimal menu or even the learner's optimal maximin value (in some sense, it relies on the optimizer to perform part of the computation). It is also essential to our approach that the maximin objective is in some sense ``one-dimensional'' -- we leave it as an interesting open question whether this technique can be generalized to the original expected utility objective.

\subsection{Related Work}

Our paper most directly builds off of the work of \cite{arunachaleswaran2024paretooptimal}, who originally introduced the framework of considering asymptotic menus of learning algorithms. This framework was introduced by \cite{arunachaleswaran2024paretooptimal} with the goal of characterizing the set of \emph{Pareto-optimal} learning algorithms. In the context of this paper, these are exactly the learning algorithms that are optimal for the learner to commit to against some distribution $\cD$ over optimizer payoffs. Surprisingly, \cite{arunachaleswaran2024paretooptimal} show that in some games $G$ (i.e., for some payoff functions $u_L$ for the learner), standard no-regret algorithms such as multiplicative weights and follow-the-regularized-leader are Pareto-dominated, implying that in these games they will never appear as solutions to our optimization problem. 

Single-shot (i.e., not repeated) instances of Stackelberg games have been somewhat well studied in the economics and computational game theory literature~\cite{Stack35}. In these games, a leader commits to an optimal strategy against a rational follower. The specific variant of the problem that we study is one where the leader has some uncertainty over the follower's type, in the form of a prior distribution over types of the follower. This problem was introduced by~\cite{conitzer2006computing}, who called it the Bayesian Stackelberg Problem. Our problem can be phrased as a Bayesian Stackelberg problem where the learner's action set is the set of all learning algorithms and the optimizer's action set is the set of all action sequences (although framing it directly in this way is not particularly helpful from a computational point of view, since both these action spaces are incredibly large). Interestingly, \cite{conitzer2006computing} show that this problem is NP-hard even for explicit bimatrix games, making it somewhat striking that we can efficiently compute the optimal no-regret algorithm to commit to. Other papers, such as ~\cite{paruchuri2008efficient} have attempted to come up with good heuristics and algorithms to solve this problem for various special cases corresponding to applications of interest, including airport security, cybersecurity, etc. 

There is also a rich line of work focused on the problem of learning the (single-shot) Stackelberg equilibrium of an unknown game via repeated interaction with an agent of some specified behavior. For example, one can try to estimate this Stackelberg equilibrium by repeatedly interacting with an agent that myopically best-responds every round \cite{letchford2009learning, blum2014learning, zhao2023online}, a no-regret agent \cite{brown2023is, goktas2022robust}, a no-swap-regret agent / an agent that best-responds to calibrated action forecasts \cite{brown2023is, haghtalab2024calibrated}, or an agent belonging to a parametrized class of behaviors \citep{sessa2020learning}. It is worth emphasizing that this is a fairly different computational problem than the problem we study in this work -- although the algorithms we commit to may ``try'' to learn the payoffs of their opponent in a similar way to these algorithms, we 1. do not make any behavioral assumptions about the opponents (allowing them to react fully rationally and non-myopically), and 2. attempt to find the Bayesian Stackelberg equilibrium of the full repeated game, not the single-shot Stackelberg equilibrium. Indeed, the recent work of \cite{ananthakrishnan2024knowledge} shows that (in some sense) there can be a gap between these two quantities; there are cases where a learner provably cannot (i.e., in no Nash equilibrium of the game) achieve their expected Stackelberg value against a distribution of opponents. 

Our paper also contributes to a growing body of research focused on the strategic manipulability of learning algorithms within game-theoretic contexts. One of the foundational contributions to this area is by \cite{braverman2018selling}, who examined these issues within the framework of non-truthful auctions involving a single buyer. Since then, similar dynamics have been explored across various economic models, such as alternative auction settings \cite{deng2019prior, cai2023selling, kolumbus2022and, kolumbus2022auctions}, contract theory \cite{guruganesh2024contracting}, Bayesian persuasion \cite{chen2023persuading}, general games \cite{deng2019strategizing, brown2023is}, and Bayesian games \cite{MMSSbayesian}. Both \cite{deng2019strategizing} and \cite{MMSSbayesian} establish that no-swap-regret is a necessary criterion to prevent optimizers from gaining by manipulating learners. 


We crucially use the existence of efficient no-swap regret algorithms in constructing our optimal no-regret algorithm in Section~\ref{sec:nr}. The first no-swap-regret algorithms were introduced by \cite{foster1997calibrated}, who also demonstrated that their dynamics converge to correlated equilibria. Since then, numerous researchers have developed learning algorithms aimed at minimizing swap regret in games \cite{hart2000simple, blum2007external, peng2023fast, dagan2023external}. Our polynomial-time algorithm for designing general menus for the maximin objective is similar to (and partially motivated by) the technique of \emph{semi-separation} recently introduced by \cite{daskalakis2024efficient}, who use this technique to design learning algorithms that minimize a variant of swap regret (linear swap regret) despite this variant being computationally hard to compute.

\section{Model and Preliminaries}\label{sec:model}





\paragraph{Notation}
Unless otherwise specified, all norms $\norm{x}$ of a vector $x \in \Rset^d$ refer to the Euclidean norm. Given any convex set $K$, we write $K^{\eps} = \{x \in \Rset^{d} \mid \exists\, y \in K \text{ s.t. } \norm{x - y} \leq \eps\}$ to denote the $\eps$-expansion of $K$. We write $\Delta_{d}$ to denote the $d$-simplex $\{ x \in \Rset^{d} \mid x_i \geq 0, \sum x_i = 1\}$.

\subsection{Online Learning}

We consider a setting where two players, a \emph{learner} $L$ and an \emph{optimizer} $O$, repeatedly play a two-player bimatrix game $G$ for $T$ rounds. The game $G$ has $m$ pure actions for the learner and $n$ pure actions for the optimizer. The learner has a fixed payoff function $u_{L}: [m] \times [n] \rightarrow [-1, 1]$, known to both parties. The optimizer has one of $k$ finitely many possible payoff functions $u_{O,i}: [m] \times [n] \rightarrow [-1, 1]$, where their type $i \in [k]$ is drawn from a known distribution $\optdist$ once at the beginning of the game. We will let $\alpha_{i} \in [0, 1]$ denote the probability associated with type $i$. 

During each round $t$, the learner picks a mixed strategy $x_t \in \Delta_{m}$ while the optimizer simultaneously picks a mixed strategy $y_{t} \in \Delta_{n}$; the learner then receives reward $u_{L}(x_t, y_t)$, and the optimizer with private type $i \in [k]$ receives reward $u_{O,i}(x_t, y_t)$ (where here we have linearly extended $u_{L}$ and $u_{O,i}$ to bilinear functions over the domain $\Delta_{m} \times \Delta_{n}$). Both the learner and optimizer observe the full mixed strategy of the other player (the ``full-information'' setting). For any mixed strategy $y \in \Delta_n$ for the optimizer, we define $\BR_{L}(y) = \arg\max_{x \in \Delta_m} u_L(x, y)$ to be the set of best responses for the learner.

The learner will employ a \emph{learning algorithm} $\cA$ to decide how to play. For our purposes, a learning algorithm is a family of horizon-dependent algorithms $\{A^T\}_{T \in \mathbb{N}}$. Each $A^{T}$ describes the algorithm the learner follows for a fixed time horizon $T$. Each horizon-dependent algorithm is a mapping from the history of play to the next round's action, denoted by a collection of $T$ functions $A^T_1, A^T_2 \cdots A^T_T$, each of which deterministically map the transcript of play (up to the corresponding round) to a mixed strategy to be used in the next round, i.e., $x_t = A^T_t(y_1, y_2,\cdots, y_{t-1})$. 

We will be interested in settings where a learner (not knowing $u_O$, but knowing $u_L$ and $\optdist$) commits to a specific learning algorithm at the beginning of the game. This algorithm is then revealed to the optimizer, who will approximately best-respond by selecting a sequence of actions that approximately (up to sublinear $o(T)$ factors, breaking ties in the learner's favor) maximizes their payoff.

\subsection{Menus}
\label{sec:menus}

This faces the learner with the very complex task of \emph{optimization over all learning algorithms} -- complex mathematical objects described by families of functions. One main contribution of this paper is to demonstrate that this optimization task is (in many cases) feasible, by following a framework introduced in \cite{arunachaleswaran2024paretooptimal} that allows us to translate statements about learning algorithms to statements about convex sets called \emph{menus}.

Note that every transcript of play of two players induces a distribution over pairs of actions they played. In particular, each transcript $\{(x_1,y_1), (x_2, y_2), \dots, (x_T,y_T) \}$ corresponds to an $mn$-dimensional \emph{correlated strategy profile (CSP)} $\phi \in \Delta_{mn}$, that is defined as 
 \[ \phi = \frac{1}{T} \sum_{t=1}^T x_t \otimes y_t. \]
 
A CSP $\phi$ contains enough information to answer many questions about the outcome of this game: for example, with the CSP $\phi$ we can compute the average utility $u_L(\phi) = \frac{1}{T} \sum_{t=1}^{T} u_L(x_t, y_t)$ of the learner (and likewise, of any of the types of optimizers). 

For this reason, instead of understanding the full complexities of a learning algorithm $\A$, it is sufficient to understand which CSPs it is possible for an optimizer to implement against this learning algorithm. Informally, this is exactly what a menu is: the \emph{asymptotic menu} (or simply, \emph{menu}) $\M(\A) \subseteq \Delta_{mn}$ is a closed, convex subset of CSPs with the property that any $\phi \in \M(\A)$ is close to some CSP that the optimizer can induce by playing some sequence of actions $y_1, y_2, \dots, y_T$ against this learning algorithm, for sufficiently large $T$.

For completeness, we spell out the exact correspondence between learning algorithms and menus (following \cite{arunachaleswaran2024paretooptimal}) in Appendix~\ref{app:algtomenus}. For the remainder of the paper, however, instead of thinking about the setting of algorithmic commitment described previously, we will instead consider the following equivalent setting of \emph{menu commitment}:

\begin{enumerate}
    \item The learner commits to a menu $\cM \subseteq \Delta_{mn}$ which is a valid asymptotic menu of some learning algorithm $\cA$.
    \item The optimizer (of type $i$) then chooses the CSP $\phi \in \cM$ that maximizes their utility $u_{O,i}(\phi)$, breaking ties in favor of the learner.
    \item The learner then receives utility $u_L(\phi)$.
\end{enumerate}

For any asymptotic menu $\M$, define $V_{L}(\M, u_O)$ to be the utility the learner ends up with under this process. Specifically, define $V_{L}(\M, u_O, \varepsilon) = \max \{ u_{L}(\csp) \mid \csp \in \cM, u_O(\csp) \ge \max_{\csp' \in \M} u_{O}(\csp') - \varepsilon\}$ and $V_{L}(\M, u_O) = V_{L}(\M, u_O, 0)$. 
We refer to any CSP in the set $\argmax_{\csp \in \cM} \{ u_{L}(\csp) \mid u_O(\csp) \ge \max_{\csp \in \M} u_{O}(\csp) \}$ as the point picked or selected by the optimizer $u_O$ in the menu $\cM$. 

The menu commitment problem above is complicated by the fact that learner must commit to a menu $\cM$ which is the asymptotic menu of some learning algorithm $\cA$ (we will often refer to such sets as \emph{valid menus}); without this constraint, the learner could simply pick their favorite CSP in all of $\Delta_{mn}$ and force the optimizer to play it. Luckily, it is possible to provide an extremely precise characterization of all possible valid menus.

\begin{theorem}[~\cite{arunachaleswaran2024paretooptimal}]\label{thm:characterization}
A closed, convex subset $\M \subseteq \Delta_{mn}$ is a valid menu iff for every $y \in \Delta_{n}$, there exists a $x \in \Delta_{m}$ such that $x \otimes y \in \M$. 
\end{theorem}

An immediate corollary of this property is that menus are upwards closed under inclusion.
\begin{corollary}
\label{corollary:upwards_closed}
    If $\cA \subseteq \cB \subseteq \Delta_{mn}$ and $\cA$ is a valid menu, then $\cB$ is also a valid menu.
\end{corollary}

In fact, an artifact of the proof of Theorem~\ref{thm:characterization} is that any menu with an efficient representation can be implemented via an efficient learning algorithm that asymptotically realizes the original menu. The main idea of this algorithm is to run an algorithm for Blackwell approachability~\cite{blackwell1956analog} to converge to the appropriate set of CSPs, combining this with a ``padding'' argument to guarantee that every point in this set is indeed approachable. For the sake of completeness, we outline the same sequence of ideas in Appendix~\ref{app:explicit_algorithms}.

\subsection{Menus and Regret}

Many learning algorithms are designed to possess specific counterfactual regret guarantees. Two of these properties -- no-regret and no-swap-regret -- are of particular importance in this paper.

A learning algorithm $\cA$ is a \emph{no-regret algorithm} if it is the case that, regardless of the sequence of actions $(x_1, x_2, \dots, x_T)$ taken by the optimizer, the learner's utility satisfies:

\begin{equation*}
\sum_{t=1}^{T} u_{L}(x_t, y_t) \geq \left(\max_{x^* \in [m]} \sum_{t=1}^{T} u_{L}(x^*, y_t)\right) - o(T).
\end{equation*}

A learning algorithm $\cA$ is a \emph{no-swap-regret algorithm} if it is the case that, regardless of the sequence of actions $(x_1, x_2, \dots, x_T)$ taken by the optimizer, the learner's utility satisfies:

\begin{equation*}
\sum_{t=1}^{T} u_{L}(x_t, y_t) \geq \max_{\pi: [m] \rightarrow [m]} \sum_{t=1}^{T} u_{L}(\pi(x_t), y_t) - o(T).
\end{equation*}

\noindent
Here the maximum is over all swap functions $\pi: [m] \rightarrow [m]$ (extended linearly to act on elements $x_t$ of $\Delta_m$). It is a fundamental result in the theory of online learning that both no-swap-regret algorithms and no-regret algorithms exist (see \cite{cesa2006prediction}). 

The above properties on algorithms directly translate into properties of the menus that they induce. We say that a given CSP $\csp$ is \emph{no-regret} if it satisfies the no-regret constraint

\begin{equation}\label{eq:no-regret-constraint}
\sum_{i\in[m]}\sum_{j\in[n]} \csp_{ij}u_{L}(i, j) \geq \max_{i^{*} \in [m]}\sum_{i\in [m]}\sum_{j\in[n]} \csp_{ij}u_{L}(i^*, j).
\end{equation}

\noindent
Similarly, we say that the CSP $\csp$ is \emph{no-swap-regret} if, for each $j \in [n]$, it satisfies

\begin{equation}\label{eq:no-swap-regret-constraint}
\sum_{i \in [m]} \csp_{ij}u_{L}(i, j) \geq \sum_{j\in[n]}\max_{i_j^{*} \in [m]} \sum_{i\in [m]}\csp_{ij}u_{L}(i_j^*, j).
\end{equation}

For a fixed $u_L$, we will define the \emph{no-regret menu} $\M_{NR}$ to be the convex hull of all no-regret CSPs, and the \emph{no-swap-regret menu} $\M_{NSR}$ to be the convex hull of all no-swap-regret CSPs. The following result states that the asymptotic menu of any no-(swap-)regret algorithm is contained in the no-(swap-)regret menu.

\begin{theorem}[\cite{arunachaleswaran2024paretooptimal}]
\label{thm:nr_nsr_containment}
If a learning algorithm $\cA$ is no-regret, then $\M(\cA) \subseteq \M_{NR}$. If $\cA$ is no-swap-regret, then $\M(\cA) \subseteq \M_{NSR}$.
\end{theorem}

Note that both $\M_{NR}$ and $\M_{NSR}$ themselves are valid asymptotic menus, since for any $y \in \Delta_{n}$, they contain some point of the form $x \otimes y$ for some $x \in \BR_{L}(y)$ (thus satisfying the characterization in Theorem \ref{thm:characterization}). Additionally, we can express each of these menus concisely as a polytope with a polynomial number of half-space constraints. We make critical use of the following properties (proved in~\cite{arunachaleswaran2024paretooptimal}) of all no-swap-regret algorithms and the set $\M_{NSR}$.

First, we see that all no-swap-regret algorithms have the same asymptotic menu, which is exactly the set $\M_{NSR}$.

\begin{theorem}[~\cite{arunachaleswaran2024paretooptimal}]\label{thm:unique_nsr_menu}
If $\cA$ is a no-swap-regret algorithm, then $\M(\cA) = \M_{NSR}$.
\end{theorem}

To complement this property, we see that the menu of every no-regret algorithm weakly contains  the set $\M_{NSR}$.

\begin{lemma}[~\cite{arunachaleswaran2024paretooptimal}]
\label{lem:nsr_within_all_nr}
    For any no-regret algorithm $\A$, $\M_{NSR} \subseteq \M(\A)$.
\end{lemma}

\sloppy{Theorem~\ref{thm:unique_nsr_menu} is significant, since there exist several known efficient no-swap-regret algorithms (see ~\cite{blum2007external},~\cite{dagan2023external},~\cite{peng2023fast}), and so we obtain an explicit algorithm which has the menu $\M_{NSR}$. Taken together with Lemma~\ref{lem:nsr_within_all_nr}, this means that we can use any no-swap-regret algorithm as a ``fallback subroutine" when constructing an optimal no-regret algorithm against some distribution of optimizers -- we make use of this idea in the proof of Theorem~\ref{thm:proof_of_nr_alg}.}

\subsection{Stackelberg Equilibria}

Finally, we present an alternative view of the polytope $\M_{NSR}$ in terms of its extreme points rather than its defining hyperplanes.

\begin{lemma}[~\cite{arunachaleswaran2024paretooptimal}]\label{lem:nsr_characterization}
The no-swap-regret menu $\M_{NSR}$ is the convex hull of all CSPs of the form $x \otimes y$, with $x \in \Delta_{m}$ and $y \in \BR_{L}(x)$.
\end{lemma}

Correlated strategy profiles of the form $x \otimes y$, with $x \in \Delta_{m}$ and $y \in \BR_{L}(x)$ have a special significance in games. They are the candidate solutions for the Stackelberg Equilibrium problem~\cite{Stack35}, which we define below.

Two players with action sets $\cB_1$ and $\cB_2$ and payoff functions $u_1$ and $u_2$ (respectively) play a Stackelberg game in the following manner. Player 1, the ``leader'', commits to an action $a^*$ in $\cB_1$ and player 2,  the ``follower'', plays a best-response $\BR(a^*) \in \cB_2$ that maximizes their payoff, tie-breaking in favor of the leader. The leader must pick an action $a^*$ to maximize their own resulting payoff. Formally

\[ a^* \in \argmax_{a \in \cB_1} u_1(a,b^*) \] 

\noindent
where $b^* \in \argmax_{b \in \cB_2} u_2(a^*,b)$ and $u_1(a^*,b^*)  \ge u_1(a^*,b')$ for all $b' \in \argmax_{b \in \cB_2} u_2(a^*,b)$. The resulting action pair $(a^*,b^*)$ is referred to as a \emph{Stackelberg equilibrium} of the game and the payoff of the leader is referred to as the \emph{Stackelberg leader value}.

In the normal form game version of this problem (where $\cB_1$ and $\cB_2$ are both simplices and the payoffs are bilinear functions), Stackelberg equilibria are guaranteed to exist and can be computed efficiently.

\begin{theorem}[\cite{conitzer2006computing}]
\label{thm:con_stack_alg}
There exists an efficient (polynomial time in the size of the game) algorithm to compute Stackelberg equilibria in normal form games.
\end{theorem}

This definition of Stackelberg equilibria allows us to reinterpret Lemma~\ref{lem:nsr_characterization} in the following manner -- the extreme points of $\M_{NSR}$ are exactly the set of candidate Stackelberg equilibria in the Stackelberg variant of our game where the optimizer is the leader and the learner is the follower. In other words, an asymptotic best-response for an optimizer playing against a learner running a no-swap-regret algorithm is to play their Stackelberg leader strategy on every single round.

\subsection{Main Results}\label{sec:results}

We conclude this section by restating our main problem in the language of menus. Note that the expected payoff of a menu $\M$ chosen by the learner against a distribution $\optdist$ of optimizers is given by

\begin{equation}\label{eq:menu_learner_value}
V_{L}(\M, \optdist) = \mathbb{E}_{u_O \sim \optdist} [ V_{L}(\M, u_O) ] = \sum_{i=1}^k \alpha_i V_{L}(\M, u_{O,i}) .
\end{equation}

We can similarly define $V_{L}(\M, \optdist, \varepsilon)$ in terms of the $V_{L}(\M, u_{O, i}, \varepsilon)$. We primarily consider the following two computational questions:

\begin{enumerate}
    \item \textbf{(General commitment)}. Given $u_L$, a distribution $\optdist$ over $u_{O,i}$, and a parameter $\varepsilon$, compute a menu $\M$ that is $\varepsilon$-approximately optimal; i.e., such that 

    \[ V_{L}(\cM^*, \optdist, \varepsilon) \ge  \max_{\text{$\M$ is a valid menu}}  V_{L}(\M, \optdist) - \varepsilon\]
    \item \textbf{(No-regret commitment)}. The same question, but with the added constraint that $\M \subseteq \M_{NR}$; that is, what is the best menu to commit to if you also want to guarantee low regret against any possible opponent (not necessarily in the support of $\optdist$)?
\end{enumerate}


The remainder of this paper is structured as follows. In Section~\ref{sec:nr}, we prove that we can efficiently (in polynomial time in $n$, $m$, and $k$) solve the problem of no-regret commitment even for $\varepsilon = 0$. In Section~\ref{sec:general}, we study the problem of general commitment. We prove that if either $k$ (the support of $\optdist$) or $mn$ (the size of the game) is constant, it is possible to solve the problem of general commitment in time polynomial in $1/\varepsilon$ and the other parameters. On the other hand, we show that it is in general hard to even decide whether a given convex set is a menu, even if it is specified concisely as an intersection of a small number of half-spaces. In particular, we show that this problem is hard even when the convex sets are restricted to be of a certain simple class associated with optimal commitments. Finally, in Appendix~\ref{app:explicit_algorithms}, we show how to ``invert" these optimal menus to derive optimal algorithms.

Our results also hold for some generalizations of these problems with minimal modification. We discuss two orthogonal generalizations below.

\begin{itemize}
    \item The optimizer's type also affects the learner's payoff. In particular, the learner now has $k$ payoff functions $\{u_{L,i}\}_{i \in [k]}$ and receives payoff $u_{L,i}(a,b)$ when playing action $a$ against the optimizer type $i$'s action $b$.
    \item The objective function of interest is not the expected payoff of the learner (over draws of the optimizer), but the minimum payoff of the learner against any possible optimizer type. More generally, any objective function that is a concave function of the CSP can be maximized.
\end{itemize}

Both generalizations only entail a small change in the objective functions used optimization subroutines in our algorithms for the no-regret and general menu commitment problem. We discuss the precise modifications in the appropriate sections.

\section{Optimal No-Regret Commitment 
}

\label{sec:nr}

We begin by discussing the problem of \emph{no-regret commitment}: the problem of finding the optimal no-regret menu $\M$ (and hence the optimal no-regret algorithm) for the learner to commit to that maximizes the expected utility of the learner. Interestingly, we will show that it is possible to solve this problem exactly and efficiently.

Recall that we assume that the optimizer is one of $k$ different types, drawn from a known distribution $\cD$ where type $i$ has probability $\alpha_i$ of occurring and corresponds to the payoff function $u_{O, i}$. Under the optimal menu $\cM$, each of these different optimizers will pick some CSP in $\cM$.  We call the $k$-tuple of CSPs $\Phi = (\csp_1, \csp_2, \dots, \csp_k) \in \Delta_{mn}^{k}$ a valid \emph{CSP assignment} of this menu $\cM$ if $\csp_i \in \arg \max_\csp u_{O,i}(\csp)$ for all $i$. A given menu may have multiple valid CSP assignments due to ties. To resolve this, we assume that the optimizer breaks ties in favor of the learner ( any further tie-breaking rules does not affect the learner's utility). However, we do not explicitly enforce this constraint; instead, by searching over all valid assignments, we inherently identify the optimal one from the perspective of the learner's utility.

Note that the eventual expected utility of a learner who has committed to the menu $\cM$ depends only on the CSP assignment of $\cM$; in fact, this expected utility is simply $\sum_{i} \alpha_i u_L(\csp_i)$. This motivates the following observation: instead of directly optimizing over all valid no-regret menus (which while simpler than the underlying learning algorithms, are still an infinite dimensional family of convex sets), we can optimize over the set of viable CSP assignments (which are $kmn$ dimensional objects) that correspond to valid no-regret menus. This has a similar flavor to the standard revelation principle in mechanism design, where instead of optimizing directly over the set of all mechanisms, it is often easier to optimize over the space of assignments of outcomes to players (with the caveat that such outcomes must satisfy some incentive compatibility guarantees).

To make use of this observation, we need a way to characterize which CSP assignments are in fact viable. Luckily, we can make use of the structure of no-regret menus -- specifically, the fact that all no-regret menus must contain the no-swap-regret menu -- to establish the following concrete characterization of such CSP assignments.

\begin{lemma}\label{lem:csp-assignment-char}
The CSP assignment $\Phi = (\csp_1, \csp_2, \dots, \csp_k)$ corresponds to a valid no-regret menu $\cM$ iff the following conditions hold:

\begin{enumerate}
    \item For all $i \in [k]$, $\csp_i \in \M_{NR}$ (every optimizer is assigned a no-regret CSP).
    \item For all $i, i' \in [k]$, $u_{O, i}(\csp_i) \geq u_{O, i}(\csp_{i'})$ (the optimizer of type $i$ prefers their CSP to the CSP assigned to any other type $i'$).
    \item For all $i \in [k]$ and no-swap-regret CSPs $\csp \in \M_{NSR}$, $u_{O, i}(\csp_i) \geq u_{O, i}(\csp)$ (the optimizer of type $i$ prefers their CSP to every no-swap-regret CSP).
\end{enumerate}
\end{lemma}
\begin{proof}
Note that if $\cM$ is a no-regret menu, then by Lemma \ref{lem:nsr_within_all_nr}, $\cM_{NSR} \subseteq \cM \subseteq \cM_{NR}$. If $(\csp_1, \csp_2, \dots, \csp_k)$ is the CSP assignment corresponding to this menu, then since $\csp_i \in \arg\max_{\csp \in \cM} u_{O, i}(\csp)$, this immediately implies the conditions above (in particular, both $\cM_{NSR}$ and $\csp_{i'}$ belong to $\cM$).

To prove the other direction, consider fix a CSP assignment $\Phi = (\csp_1, \csp_2, \dots, \csp_k)$ satisfying the above conditions, and consider the menu $\M = \conv(\M_{NSR} \cup \{\csp_1, \csp_2, \dots, \csp_k\})$. Since menus are upwards closed under inclusion, this is a valid menu, and since each $\csp_i \in \M_{NR}$, this is a no-regret menu ($\M \subseteq \M_{NR}$). The two latter constraints then imply that $\csp_i \in \arg\max_{\csp \in \cM} u_{O, i}(\csp)$, and thus $\Phi$ is the corresponding CSP assignment to $\cM$. 
\end{proof}

With a little more work, Lemma \ref{lem:csp-assignment-char} allows us to write the set of valid no-regret CSP assignments as an explicit polytope in $\Rset^{kmn}$. For each type $i$, let $v_{i} = \max_{\csp \in \M_{NSR}} u_{O, i}(\csp)$ be the utility of the optimizer of type $i$ under their favorite no-swap-regret CSP. Note that since the extreme points of the no-swap-regret menu are Stackelberg CSPs (Lemma \ref{lem:nsr_characterization}), this is equivalently the Stackelberg leader value of the optimizer of type $i$ and therefore can be computed efficiently. We can now write down the following linear program in the variables $\csp_{i}$ to compute the optimal CSP assignment.

\begin{align}
    & \max \sum_{i=1}^{k} \alpha_i u_L(\csp_{i}) \notag \\
    \text{s.t. } &   \csp_{i} \in \mathcal{M}_{NR}, \quad \forall i \tag{Each $\csp_i$ is a no-regret CSP for the learner} \notag \\ 
    & u_{O, i}(\csp_i) \geq u_{O, i}(\csp_{i'}), \quad \forall i, i' \in [k] \tag{The $i$-th optimizer weakly prefers CSP $\csp_i$ to any other $\csp_{i'}$} \notag \\
    & u_{O, i}(\csp_i) \geq v_i, \quad \forall i \in [k] \tag{The $i$-th optimizer values $\csp_i$ as highly as their Stackelberg leader value} \notag \\
    \label{eq:no-reg-csp-lp}
\end{align}

The constraints in the LP \eqref{eq:no-reg-csp-lp} correspond exactly to the three constraints in Lemma \ref{lem:csp-assignment-char}, and therefore also characterizes the set of all valid CSP assignments and computes the optimal no-regret CSP assignment. We summarize this process in Algorithm \ref{alg:d-opt}.

\begin{algorithm}
\caption{Algorithm to Find the Optimal No-Regret Commitment} 
\label{alg:d-opt}
\begin{algorithmic}[1] 
    \STATE \textbf{Input}: Learner utility function $u_L$
    \STATE \textbf{Input}: Distribution $\mathcal{D}$ with support $k$ over $mn$-length vectors, representing possible optimizer types/ payoffs.
    
    \STATE For each type $i \in [k]$, compute the Stackelberg leader value $v_i = \max_{\csp \in \M_{NSR}} u_{O, i}(\csp)$ (this can be done in time $\poly(m, n)$ per type by Theorem \ref{thm:con_stack_alg}).
    
    \STATE Solve the linear program \eqref{eq:no-reg-csp-lp} for the optimal CSP assignment $\Phi = (\csp_1, \csp_2, \dots, \csp_k)$.
    
    \STATE \textbf{Output}: The menu  $\cM := \conv(\mathcal{M}_{NSR} \cup \{\csp_1, \csp_2, \dots, \csp_k\})$.
\end{algorithmic}
\end{algorithm}

\begin{theorem}\label{thm:proof_of_nr_alg} Algorithm~\ref{alg:d-opt} finds the optimal no-regret menu $\M$ to commit to for a learner with payoff $u_L$ against a distribution $\cD$ of optimizers. This algorithm runs in time polynomial in the size of the game and the support of the distribution $\cD$. 
\end{theorem}
\begin{proof}
Since $v_i = \max_{\csp \in \M_{NSR}} u_{O, i}(\csp)$, the three constraints in the LP \eqref{eq:no-reg-csp-lp} correspond (in order) to the three constraints in Lemma \ref{lem:csp-assignment-char}, any solution $\Phi = (\csp_1, \csp_2, \dots, \csp_k)$ to the LP \eqref{eq:no-reg-csp-lp} is the CSP assignment for a valid menu $\M$ (in particular, the $\M$ output by Algorithm~\ref{alg:d-opt}). Since the expected utility of a learner using this menu is precisely $\sum_{i} \alpha_i u_L(\csp_i)$, the CSP assignment returned by the LP \eqref{eq:no-reg-csp-lp} is the CSP assignment of the no-regret menu that optimizes this expected utility. Finally, the algorithm runs in polynomial time -- each computation of the Stackelberg value $v_i$ takes polynomial time, and the final LP has $mnk$ variables and $O(k^2 + kn)$ constraints (verifying that $\csp_i \in \M_{NR}$ can be done by checking that the regret with respect to each of the $n$ pure learner actions is non-positive).
\end{proof}

We conclude this section with several remarks on how Theorem \ref{thm:proof_of_nr_alg} can be generalized to other settings. First, note that we can replace the expected utility objective in the LP \eqref{eq:no-reg-csp-lp} with any concave function of the variables $\csp_i$ and still efficiently solve the resulting program. This lets us capture the generalizations mentioned at the end of Section \ref{sec:results}. Namely, if we want to e.g. maximize the worst-case utility against any of the optimizers, we can replace the current objective with the concave function $\min_{i \in [k]} u_L(\csp_i)$. Similarly, if the type $i$ affects the learner's utility, we can instead consider the linear objective $\sum_{i} u_{L, i}(\csp_i)$. 

Secondly, the crucial structural property we used about no-regret menus is that every such menu $\cM$ both contains the valid menu $\cM_{NSR}$ and is contained in the valid menu $\cM_{NR}$. The procedure in Algorithm \ref{alg:d-opt} can be adjusted to optimize over any ``sandwich set'' of menus; i.e., all menus $\cM$ satisfying $\cM_{1} \subseteq \cM \subseteq \cM_2$ for given valid menus $\cM_1$ and $\cM_2$. One main difficulty in solving the general menu commitment problem is that the set of all valid menus is \emph{not} of this form.

Finally, it is natural to wonder if this procedure is necessary at all, or whether the no-swap-regret menu $\cM_{NSR}$ (the most restrictive of all no-regret menus) always leads to the best outcome for the learner. The following theorem (whose proof is deferred to Appendix \ref{app:omitted}) shows that there are cases where no-swap-regret algorithms are strictly sub-optimal.

\begin{theorem}
\label{thm:nsr_not_optimal}
    There exists learner utility $u_{L}$ and a distribution over optimizer types $\mathcal{D}$ (with support $k=1$) such that the menu output by Algorithm~\ref{alg:d-opt} is strictly better than the no-swap regret menu $\cM_{NSR}$. 
\end{theorem}

\section{Optimal General Commitment}
\label{sec:general}

We now lift the constraint that the menus we consider must correspond to no-regret algorithms, and study the problem of general menu commitment. This will turn out to be a considerably harder problem than the no-regret commitment problem studied in the previous section, and we will no longer be able to establish an efficient algorithm that exactly optimizes over the space of all possible menus. Our ultimate goal in this section is to establish the following theorem, which shows that whenever \emph{either} the size $mn$ of the game or the size $k$ of the support of the optimizer distribution $\cD$ is constant, then there exists a polynomial-time algorithm for computing an $\eps$-approximate optimal menu. 

\begin{theorem}
\label{thm:main_result_without_nr}
    Let $d = \min(mn, k)$. There exists an algorithm that runs in time 
    \begin{align*}
    \left(\frac{mnk}{\eps}\right)^{O(d)}
    \end{align*}
    and constructs an $\eps$-approximate optimal menu $\cM^\eps$. Moreover, this menu has an efficient membership oracle (running in time $\poly(m, n, k)$) and therefore corresponds to a learning algorithm with a polynomial per-round time complexity.
\end{theorem}

Our high-level strategy will be similar to the one employed previously for the problem of no-regret commitment. As in that case, instead of thinking directly about the problem of menu design, we will attempt to design the optimal CSP assignment $\Phi = (\csp_1, \csp_2, \dots, \csp_k)$ that corresponds to a valid menu. The main difference between these two settings is that we will lose the structural property present for no-regret menus that any valid no-regret menu must contain $\M_{NSR}$ as a sub-menu. In particular, we will not be able to always construct our optimal menu by taking the convex hull of our ``fallback'' menu $\M_{NSR}$ with the CSPs in our assignment $\Phi$. Instead, for any CSP assignment $\Phi$, we will have to determine whether there exists a viable fallback menu $\M'$ that can take the place of $\M_{NSR}$ in this construction.

To formalize this, given any CSP assignment $\Phi = (\csp_1, \csp_2, \dots, \csp_k)$, we define its \emph{candidate menu} $\cC(\Phi) := \{ \csp \in \Delta^{mn} \,\mid\,u_{O,i}(\csp) \le u_{O,i}(\csp_i) \,\forall i \in [k] \}$ to be the set of CSPs $\csp$ where the optimizer of type $i$ prefers their assigned CSP $\csp_i$ to $\csp$ for all types $i$. Intuitively, this is the largest possible ``fallback menu'' we could offer where each optimizer would select their assigned CSP over any CSP in this set (with the important caveat that this set may not in fact be a valid menu). In particular, we can characterize the set of valid CSP assignments $\Phi$ in the following manner.

\begin{lemma}\label{lem:csp-assignment-char-general}
The CSP assignment $\Phi = (\csp_1, \csp_2, \dots, \csp_k)$ corresponds to a valid menu $\cM$ iff the following conditions hold:

\begin{enumerate}
    \item For each $i \in [k]$, $\csp_i \in \cC(\Phi)$.
    \item The set $\cC(\Phi)$ is a valid menu.
\end{enumerate}

\noindent
If these two conditions hold, then $\Phi$ is the CSP assignment for the menu $\cC(\Phi)$.
\end{lemma}
\begin{proof}
First, note that if both conditions hold, then by the definition of $\cC(\Phi)$ the optimizer of type $i$ will prefer $\csp_i$ to any other CSP in the menu $\cC(\Phi)$, and therefore $\Phi$ is the CSP assignment for the valid menu $\cC(\Phi)$.

On the other hand, if $\Phi$ is the CSP assignment for some valid menu $\M$, note that we must have $\M \subseteq \cC(\Phi)$ (any $\csp \in \M \setminus \cC(\Phi)$ would be better for one of the optimizers than their assigned CSP, by the construction of $\cC(\Phi)$). It follows that $\cC(\Phi)$ is a valid menu, since menus are upwards closed under inclusion (Corollary \ref{corollary:upwards_closed}). Moreover, in any CSP assignment we must have $u_{O, i}(\csp_i) \geq u_{O, i}(\csp_{i'})$ for any $i, i' \in [k]$, and therefore each $\csp_i$ is guaranteed to belong to $\cC(\Phi)$. 
\end{proof}

We again wish to use this characterization to optimize over the space of valid CSP assignments. To this end, let $\cP \subseteq \Delta_{mn}^{k}$ be the set of all CSP assignments that correspond to a valid menu. We will use the characterization in Lemma~\ref{lem:csp-assignment-char-general} to write $\cP$ as the intersection of the following two convex sets. First, we have the \emph{incentive compatible set} $\cR$ containing all CSP assignments satisfying the first condition of Lemma~\ref{lem:csp-assignment-char-general}, namely,

\[ \cR := \{ \Phi \in \Delta^{mnk} \mid u_{O,i}(\csp_i) \ge u_{O,i}(\csp_{i'}) \,\forall i,i' \in [k] \}.\]

Second, we have the \emph{valid candidate set} $\cS$ containing all CSP assignments satisfying the second condition of Lemma~\ref{lem:csp-assignment-char-general}, namely,

\[ \cS := \{ \Phi \in \Delta_{mn}^{k} \mid \cC(\Phi) \text{ is a valid menu}\}\]

Note that the description of $\cR$ above implies that it is a convex polytope. Although less obvious, it can also be shown that the set $\cS$ is convex.

\begin{lemma}\label{lem:S_is_convex}
The set $\cS$ is a convex subset of $\Delta_{mn}^{k}$.
\end{lemma}
\begin{proof}
Recall (by Theorem \ref{thm:characterization}) that $\cM$ is a valid menu iff for every $y \in \Delta_n$, there exists an $x \in \Delta_{m}$ such that $x \otimes y \in \cM$. Let $\Phi_1 = (\csp_{1, 1}, \csp_{1, 2}, \dots, \csp_{1, k})$ and $\Phi_2 = (\csp_{2, 1}, \csp_{2, 2}, \dots, \csp_{2, k})$ be two elements of $\cS$. We will show for any $\Phi' = \lambda \Phi_1 + (1-\lambda)\Phi_2$ (where $\lambda \in [0,1]$) that $\cC(\Phi')$ is a valid menu and thus that $\Phi' \in \cS$. 

Fix any $y \in \Delta_n$. It suffices to show that there exists an $x \in \Delta_m$ such that $x \otimes y \in \cC(\Phi')$. Since $\Phi_1$ and $\Phi_2$ belong to $\cS$, there exist elements $x_1, x_2 \in \Delta_m$ such that $x_1 \otimes y \in \cC(\Phi_1)$ and $x_2 \otimes y \in \cC(\Phi_2)$. In particular, we have that $u_{O, i}(x_1, y) \leq u_{O, i}(\csp_{1, i})$ and $u_{O, i}(x_2, y) \leq u_{O, i}(\csp_{2, i})$ for all $i \in [k]$. But since $u_{O, i}$ is bilinear, this implies that $u_{O, i}(\lambda x_1 + (1-\lambda) x_2, y) \leq \lambda u_{O, i}(\csp_{1, i}) + (1-\lambda) u_{O, i}(\csp_{2, i}) = u_{O, i}(\lambda \csp_{1, i} + (1-\lambda)\csp_{2, i}) = u_{O, i}(\csp'_{i})$. Therefore $x' = \lambda x_1 + (1-\lambda)x_2$ has the property that $x' \otimes y \in \cC(\Phi')$, and therefore $\Phi' \in \cS$. 
\end{proof}

From Lemma~\ref{lem:csp-assignment-char-general}, it immediately follows that $\cP = \cR \cap \cS$. Note furthermore that we have an explicit (and efficiently checkable) description of $\cR$ as a convex polytope with a polynomial number of facets. If we had a similar description for $\cS$, we would be able to efficiently optimize over all of $\cP$ and solve the general menu commitment problem. Unfortunately, no such succinct description exists for $\cS$ in general (as we will later show, it is NP-hard to even decide membership in $\cS$). Instead, we will show how to construct an approximate separation oracle for $\cS$ that runs in time exponential in $\min(mn, k)$, but polynomial in all other parameters. This oracle will eventually allow us to algorithmically construct an $\epsilon$-approximate optimal menu. 

To formally establish this, we introduce the following relaxations of the above definitions. For a CSP assignment $\Phi$, let $\cC(\Phi, \eps)$ denote the $\eps$-approximate candidate menu defined via $\cC(\Phi, \eps) := \{ \csp \in \Delta^{mn} \,\mid\,u_{O,i}(\csp) \le u_{O,i}(\csp_i) + \eps \,\forall i \in [k] \}$. 

\begin{lemma}\label{lem:vl_and_c}
If $\cC(\Phi, \eps)$ is a valid menu with the property that $\csp_i \in \cC(\Phi, \eps)$ for all $i \in [k]$, then

\[
V_L(\cC(\Phi, \eps), \cD, \eps) \geq V_{L}(\Phi).
\]
\end{lemma}
\begin{proof}
\sloppy{Recall that $V_{L}(\cC(\Phi, \eps), \cD, \eps) = \sum_{i} \alpha_i V_{L}(\cC(\Phi, \eps), u_{O, i}, \eps)$. It therefore suffices to show that $V_{L}(\cC(\Phi, \eps), u_{O, i}, \eps) \geq u_{L}(\csp_i)$. But this directly follows from the definitions of $\cC(\Phi, \eps)$ and $V_{L}(\cM, \cD, \eps)$; in particular, $\csp_i$ has the property that $u_{O, i}(\csp_i) \geq \max_{\csp' \in \cC(\Phi, \eps)} u_{O,i}(\csp') - \eps$ and so it is a candidate CSP in the computation of $V_L(\cC(\Phi, \eps), u_{O, i}, \eps)$.}
\end{proof}

We likewise define the relaxed candidate set $\cS(\eps)$ via 

\[ \cS(\eps) := \{ \Phi \in \Delta_{mn}^{k} \mid \cC(\Phi, \eps) \text{ is a valid menu}\}\]

\noindent
We will work with \emph{approximate separation oracles} for $\cS$ of the following form: given as input a CSP assignment $\Phi \in \Delta_{mn}^k$ and a precision parameter $\delta$, the oracle either returns that $\Phi \in \cS(\delta)$ or returns a hyperplane separating $\Phi$ from $\cS$. This can be thought of as analogous to the standard definition of a weak separation oracle for the set $\cS$ (which either returns that $\Phi \in \cS^{\delta}$ or a separating hyperplane). 





Below, we show that such an  oracle for $\cS$ allows us to compute an $\eps$-approximate optimal menu (thus letting us eventually establish Theorem~\ref{thm:main_result_without_nr}). Here we make one additional technical assumption (common in convex optimization, and necessary for execution of cutting-plane methods such as the ellipsoid algorithm): we assume that $\cP$ contains a small ball of radius $r_{\min}$ with $\log(1/r_{\min}) = \poly(m, n, \log(1/\eps))$. Note that since we are only concerned with finding approximately optimal menus, this assumption is essentially true without loss of generality, since we can expand $\cR$ and $\cS$ by $\eps/100$ (thus guaranteeing $r_{\min} = \Omega(\eps)$) without significantly changing the approximation guarantee; nonetheless, for ease of exposition we will proceed under the assumption that this ball exists. 

\begin{lemma}\label{lem:from-weak-membership}
Assume we have access to an approximate separation oracle for the set $\cS$ (as defined above). Then there exists an algorithm that constructs an $\eps$-approximate optimal menu $\cM^{\eps}$ that makes at most $\poly(m, n, \log (1/\eps))$ oracle calls with precision $\eps$. 
\end{lemma}
\begin{proof}
Note that by additionally checking the constraints associated with the set $\cR$, we can construct an approximate separation oracle that either returns that $\Phi \in \cR \cap \cS(\eps)$ or a hyperplane separating $\Phi$ from $\cR \cap \cS = \cP$. We can use this primitive in combination with standard cutting-plane methods (see~\cite{lee2018efficient}) to find a CSP assignment $\Phi^* \in \cR \cap \cS(\eps)$ with the property that $V_{L}(\Phi^*) \geq \max_{\Phi \in \cP} V_L(\Phi) - \eps$. This procedure makes at most $\poly(m, n, \log (1/\eps),  \log (1/r_{\min})) = \poly(m, n, \log (1/\eps))$ calls to this oracle.  

We claim that the menu $\cM^{\eps} = \cC(\Phi^*, \eps)$ is an $\eps$-approximate menu. To see this, note that by Lemma \ref{lem:vl_and_c} we have that $V_{L}(\cM^{\eps}, \cD, \eps) \geq V_{L}(\Phi^*)$. By the previous guarantee, this is at least $\max_{\Phi \in \cP} V_L(\Phi) - \eps = \max_{\text{$\M$ is a valid menu}}  V_{L}(\M, \optdist) - \eps$, and therefore $\M^{\eps}$ is $\eps$-optimal.
\end{proof}

\subsection{Approximate oracles for the valid candidate set}

We now turn our attention to constructing explicit approximate separation oracles for the valid candidate set $\cS$. In combination with Lemma \ref{lem:from-weak-membership}, this will let us construct explicit algorithms for producing approximately optimal menus for the general menu commitment problem.

In this section we will demonstrate two different constructions of such an oracle: one that runs in polynomial time when the size $mn$ of the game is constant, and one that runs in polynomial time when the size $k$ of the support of the distribution $\optdist$ is constant. Both constructions will reduce the problem of determining whether a given CSP assignment corresponds to a valid menu to the problem of deciding whether a set is approachable in a specific instance of Blackwell approachability. We therefore begin by reviewing the relevant tools we need from the theory of Blackwell approachability.

The core object of study in Blackwell approachability is a \emph{vector-valued game}: a game where two players take actions in some convex set and receive a ``payoff'' in the form of a bilinear vector-valued function of their two actions. In all of our applications of Blackwell approachability, we will continue to consider games between a learner and an optimizer where the learner plays actions $x \in \Delta_m$ and the optimizer plays actions $y \in \Delta_n$. However, instead of considering the specific (single-dimensional) payoff functions $u_L$ and $u_{O, i}$, we will consider an arbitrary bilinear \emph{$d$-dimensional vector-valued} payoff function $u: \Delta_m \times \Delta_n \rightarrow \Rset^{d}$. 

The goal of Blackwell approachability is to characterize which $d$-dimensional sets $\cM$ the learner can force the time-averaged value of $u(x_t, y_t)$ to approach. The main theorem of Blackwell approachability shows that these sets are precisely the \emph{response-satisfiable} sets.

\begin{definition}[Response-Satisfiability]
\label{def:response_satisfiability}
Let \( \cM \subseteq \mathbb{R}^d \) be a closed convex set, and let \( u: \Delta_m \times \Delta_n \rightarrow \mathbb{R}^d \) be a vector-valued payoff function. The set \( \cM \) is said to be \emph{response-satisfiable} with respect to \( u \) if for every mixed strategy \( y \in \Delta_m \) of the optimizer, there exists a mixed strategy \( x \in \Delta_n \) of the learner such that
\[
u(x, y) \in \cM.
\]
\end{definition}

\begin{theorem}[Blackwell Approachability Theorem]\label{thm:blackwell}
Let \( \cM \subseteq \mathbb{R}^d \) be a closed convex set, and let \( u: \Delta_m \times \Delta_n \rightarrow \mathbb{R}^d \) be a bilinear vector-valued payoff function. If \( \cM \) is response-satisfiable with respect to \( u \), then there exists a learning algorithm \( \mathcal{A} \) for the learner such that, for any sequence of strategies \( y_1, y_2, \dotsc \) chosen by the optimizer, the sequence of the learner's strategies \( x_1, x_2, \dotsc \) produced by \( \mathcal{A} \) satisfies
\[
\lim_{T \rightarrow \infty} d\left( \frac{1}{T} \sum_{t=1}^T u(x_t, y_t), \cM \right) = 0,
\]
where \( d\left( \mathbf{v}, \cM \right) = \inf_{\mathbf{m} \in \cM} \| \mathbf{v} - \mathbf{m} \| \) denotes the Euclidean distance from vector \( \mathbf{v} \) to the set \( \cM \).
\end{theorem}

Blackwell approachability is closely related to the validity of menus in the following sense. Consider the instance of Blackwell approachability where $d = mn$ and $u(x, y) = x \otimes y$ (i.e., the single round CSP where the learner plays $x$ and the optimizer plays $y$). Then a set $\cM \subseteq \Delta_{mn}$ is a valid asymptotic menu iff it is approachable in this instance of Blackwell approachability (indeed, this is just a restatement of Theorems \ref{thm:characterization} and \ref{thm:blackwell}). In particular, to decide whether a CSP assignment $\Phi$ belongs to $\cS$, it suffices to decide whether the set $\cC(\Phi)$ is approachable.

Our main tool to decide whether a given set is approachable is the following theorem of \cite{mannor2009approachability}, which runs in time exponential in $d$ (but polynomial in all other parameters) and approximately decides whether a set is approachable.

\begin{theorem}[~\cite{mannor2009approachability}]
\label{thm:tsitsiklis}
    There is an algorithm which, given a candidate menu $\M \subseteq \Delta_{mn}$ (as an explicit polytope with a polynomial number of defining faces) and a precision parameter $\delta$, runs in time $(mnk/\delta)^{O(d)}$ and correctly either:
    \begin{itemize}
        \item Outputs that the set $\M^{\delta}$ is approachable
        \item Outputs that $\M$ is not approachable, along with a $y \in \Delta_{n}$ certifying that $\M$ is not response satisfiable (i.e., a $y$ s.t. $u(x, y) \not\in \M$ for all $x \in \Delta_n$).
    \end{itemize}
\end{theorem}

To use Theorem \ref{thm:tsitsiklis} to construct our approximate separation oracle, we will need two auxiliary lemmas. First, we will need to relate the property of (some expansion of) $\cC(\Phi)$ being approachable to the property of $\Phi$ belonging to $\cS(\eps)$. 

\begin{lemma}\label{lem:csp-expand}
Let $\Phi$ be a CSP assignment. If $\cC(\Phi)^{\eps}$ is a valid menu, then $\Phi \in \cS(\eps\sqrt{mn})$. 
\end{lemma}
\begin{proof}
Note that $\cC(\Phi)^{\eps} \subseteq \cC(\Phi, \eps\sqrt{mn})$, since if $\norm{\csp - \csp'} \leq \eps$, it must be the case that $|u_{O, i}(\csp) - u_{O, i}(\csp')| \leq \sqrt{mn}$ (since all coefficients of $u_{O, i}$ are bounded in $[-1, 1]$). Since menus are upwards closed under inclusion, this means that $\cC(\Phi, \eps\sqrt{mn})$ is a valid menu, and therefore $\Phi \in \cS(\eps\sqrt{mn})$.
\end{proof}

Secondly, we will show that a certificate of non-response-satisfiability can be used to construct a separating hyperplane between $\Phi$ and $\cS$.

\begin{lemma}\label{lem:separate}
Assume we are given a CSP assignment $\Phi = (\csp_1, \csp_2, \dots, \csp_k)$ and a $y \in \Delta_n$ such that $x \otimes y \not\in \cC(\Phi)$ for any $x \in \Delta_m$. Then it is possible (in $\poly(m, n, k)$ time) to construct a $kmn$-dimensional hyperplane separating $\Phi$ from $\cS$.
\end{lemma}
\begin{proof}
Let $\cC_{\text{sep}}$ denote the convex set $\Delta_{m} \otimes y$. By assumption, we have that $\cC_{\text{sep}}$ and $\cC(\Phi)$ are disjoint convex sets. By the theory of LP duality, there must be a hyperplane separating $\cC(\Phi)$ from $\cC_{\text{sep}}$ that can be written as a linear combination of the defining halfspaces of $\cC(\Phi)$. Since $\cC(\Phi)$ is defined by halfspaces of the form $u_{O, i}(\csp) \leq u_{O, i}(\csp_i)$, this means that there is a vector $h \in \Rset_{\geq 0}^k$ such that the inequality

\begin{equation}\label{eq:sep1}
\sum_{i=1}^{k} h_{i}u_{O, i}(\csp) \leq \sum_{i=1}^{k} h_{i}u_{O, i}(\csp_i)
\end{equation}

\noindent
holds for all $\csp \in \cC(\Phi)$ and is violated for every $\csp \in \cC_{\text{sep}}$.

Now consider an arbitrary CSP assignment $\Phi' = (\csp'_1, \csp'_2, \dots, \csp'_k) \in \Delta_{mn}^k$. Consider the half-space constraint (in $kmn$-dimensional space) implied by the inequality

\begin{equation}\label{eq:sep2}
\sum_{i=1}^{k} h_{i}u_{O, i}(\csp'_i) \leq \sum_{i=1}^{k} h_{i}u_{O, i}(\csp_i)
\end{equation}

Note that $\Phi$ clearly satisfies this constraint. On the other hand, every $\Phi' \in \cS$ must have the property that $\cC(\Phi')$ contains an element of the form $x \otimes y$ (for some $x$), and thus must have non-empty intersection with $\cC_{\text{sep}}$. But any CSP $\csp \in \cC(\Phi')$ must satisfy 

\begin{equation}\label{eq:sep3}
\sum_{i=1}^{k} h_{i}u_{O, i}(\csp) \leq \sum_{i=1}^{k} h_{i}u_{O, i}(\csp'_i)
\end{equation}

\noindent
and therefore (by \eqref{eq:sep2}) also satisfies inequality \eqref{eq:sep1}, and thus cannot lie in $\cC_{\text{sep}}$. It follows that the hyperplane given by \eqref{eq:sep2} is a hyperplane separating $\Phi$ from $\cS$.
\end{proof}

We are now ready to construct our first separation oracle (that is efficient for small values of $mn$).

\begin{theorem}\label{thm:small_mn}
There exists an algorithm that takes as input a CSP assignment $\Phi$ and a precision parameter $\delta$, runs in time $(mnk/\delta)^{O(mn)}$ and either outputs that $\Phi \in \cS(\delta)$ or returns a hyperplane separating $\Phi$ from $\cS$.
\end{theorem}
\begin{proof}
We run the algorithm of \cite{mannor2009approachability} (Theorem \ref{thm:tsitsiklis}) on the set $\cC(\Phi)$ with precision $\delta' = \delta/\sqrt{mn}$. If this algorithm returns that $\cC(\Phi)^{\delta'}$ is approachable, then by Lemma \ref{lem:csp-expand}, we are guaranteed that $\Phi \in \cS(\delta)$, so we return that. On the other hand, if this algorithm returns a $y \in \Delta_{n}$ certifying that $\cC(\Phi)$ is not response-satisfiable, we follow the procedure in Lemma \ref{lem:separate} to construct a hyperplane separating $\Phi$ from $\cS$ and return that. The time of this procedure is dominated by the $(mnk/\delta)^{O(mn)}$ time it takes to run the procedure of \cite{mannor2009approachability}.
\end{proof}

\subsubsection{Small support}

To handle the case when $k$ is small, we consider a slightly different reduction to the problem of Blackwell approachability with the $k$-dimensional bilinear payoff function $u:\Delta_{m} \times \Delta_{n}: \Rset^{k}$ defined via $u(x, y) = (u_{O, 1}(x, y), u_{O, 2}(x, y), \dots, u_{O, k}(x, y))$. The key observation is that for any CSP $\csp \in \Delta_{mn}$ the $k$ values in $u(\csp)$ are sufficient for determining whether $\csp$ belongs to the candidate set $\cC(\Phi)$ for a given CSP assignment $\Phi$. In particular, for any CSP assignment $\Phi$, define the \emph{candidate utility set} $\cU(\Phi)$ via

\[\cU(\Phi) := \{u \in \Rset^k \mid u_i \leq u_{O, i}(\csp_i) \, \forall i \in [k]\}.\]

\noindent
Likewise, define the relaxed candidate utility set $\cU(\Phi, \eps)$ via

\[\cU(\Phi) := \{u \in \Rset^k \mid u_i \leq u_{O, i}(\csp_i) + \eps \, \forall i \in [k]\}.\]

The following lemma formally establishes the above analogy between the sets $\cU$ and the menus $\cC$.

\begin{lemma}\label{lem:util-analogy}
Fix a CSP assignment $\Phi$, a CSP $\csp$, and any $\eps \geq 0$. Then $\csp \in \cC(\Phi, \eps)$ iff $u(\csp) \in \cU(\Phi, \eps)$. Moreover, $\cC(\Phi, \eps)$ is a valid menu iff $\cU(\Phi, \eps)$ is approachable.  
\end{lemma}
\begin{proof}
If $\csp \in \cC(\Phi, \eps)$, then $u(\csp)_i = u_{O, i}(\csp) \leq u_{O, i}(\csp_i) + \eps$, and therefore $u(\csp)$ satisfies all the constraints of $\cU(\Phi, \eps)$. Likewise, if $u(\csp) \in \cU(\Phi, \eps)$, then $\csp$ satisfies all the constraints of $\cC(\Phi, \eps)$.

Recall that $\cC(\Phi, \eps)$ is a valid menu iff it is response-satisfiable: for any $y \in \Delta_n$ there exists an $x \in \Delta_m$ such that $x \otimes y \in \cC(\Phi, \eps)$. But $x \otimes y \in \cC(\Phi, \eps)$ iff $u(x, y) \in \cU(\Phi, \eps)$, and so it follows that $\cU(\Phi, \eps)$ is response-satisfiable (and thus approachable) iff $\cC(\Phi, \eps)$ is.  
\end{proof}

\noindent
We additionally have the following analogue of Lemma \ref{lem:csp-expand} for expansions of $\cU(\Phi)$.

\begin{lemma}\label{lem:csp-expand-util}
Let $\Phi$ be a CSP assignment. If $\cU(\Phi)^{\eps}$ is an approachable set, then $\Phi \in \cS(\eps)$. 
\end{lemma}
\begin{proof}
It suffices to show that $\cU(\Phi)^{\eps} \subseteq \cU(\Phi, \eps)$, since if $\cU(\Phi, \eps)$ is approachable then $\cC(\Phi, \eps)$ is a valid menu, and therefore $\Phi \in \cS(\eps)$. Let $u' \in \cU(\Phi)^{\eps}$ and $u \in \cU(\Phi)$ be so that $\norm{u - u'} \leq \eps$. But then $|u_i - u'_i| \leq \eps$ for any coordinate $i \in [k]$, and it follows that $u' \in \cU(\Phi, \eps)$. 
\end{proof}

We can now construct our second separation oracle.

\begin{theorem}\label{thm:small_k}
There exists an algorithm that takes as input a CSP assignment $\Phi$ and a precision parameter $\delta$, runs in time $(mnk/\delta)^{O(k)}$ and either outputs that $\Phi \in \cS(\delta)$ or returns a hyperplane separating $\Phi$ from $\cS$.
\end{theorem}
\begin{proof}
Similar to the proof of Theorem \ref{thm:small_mn}, we run the algorithm of \cite{mannor2009approachability} (Theorem \ref{thm:tsitsiklis}) on the set $\cU(\Phi)$ with precision $\delta$. If this algorithm returns that $\cU(\Phi)^{\delta}$ is approachable, then by Lemma \ref{lem:csp-expand-util}, we are guaranteed that $\Phi \in \cS(\delta)$, so we return that. On the other hand, if this algorithm returns a $y \in \Delta_{n}$ certifying that $\cU(\Phi)$ is not response-satisfiable, then $\cC(\Phi)$ is likewise not response-satisfiable (by Lemma \ref{lem:util-analogy}), and we can again follow the procedure in Lemma \ref{lem:separate} to construct a hyperplane separating $\Phi$ from $\cS$ and return that. The time of this procedure is dominated by the $(mnk/\delta)^{O(k)}$ time it takes to run the procedure of \cite{mannor2009approachability}.
\end{proof}

Finally, by combining Lemma \ref{lem:from-weak-membership} with the better of these two separation oracles, we can prove Theorem \ref{thm:main_result_without_nr}.

\begin{proof}[Proof of Theorem~\ref{thm:main_result_without_nr}]
If $mn < d$, run the reduction of Lemma \ref{lem:from-weak-membership} with the separation oracle defined in Theorem \ref{thm:small_mn}. Otherwise, run the same reduction but with the separation oracle defined in Theorem \ref{thm:small_k}. The total time complexity is at most $(mnk/\eps)^{O(\min(mn, k))}$, as desired.
\end{proof}

\subsection{Hardness of testing the validity of menus}

While our algorithm finds an optimal commitment efficiently when either $k$ or $mn$ is small, our procedure does not run in time polynomial in the number of actions and the number of types simultaneously. In this section we show that our algorithmic approach would require significant alterations to do so by demonstrating that it is in general hard to decide whether a given CSP assignment can correspond to a valid menu (i.e., decide membership in $\cS$).

\begin{theorem}
\label{thm:hardness}
The problem of deciding whether $\cC(\Phi)$ is a valid menu (given as input a CSP assignment $\Phi \in \Delta_{mn}^k$ and a collection of optimizer payoffs $u_{O, i}$) is NP-hard.
\end{theorem}

The proof follows from a similar hardness result in~\cite{mannor2009approachability} about the hardness of testing approachability of the non-negative orthant. We state this result below.

\begin{theorem}[~\cite{mannor2009approachability}]
\label{thm:tsitsiklis_hardness}
    The problem of taking as input a vector-valued function $u: \Delta_{m} \times \Delta_{n} \rightarrow \Rset^{d}$ and deciding whether the non-negative orthant $[0, \infty)^{d}$ is approachable is NP-hard.
\end{theorem}

To prove Theorem~\ref{thm:hardness}, we will demonstrate how to reduce the approachability problem in Theorem~\ref{thm:tsitsiklis_hardness} to an instance of our problem of determining whether a menu is valid. The proof is deferred to Appendix~\ref{app:omitted}.


\subsection{A Polynomial Time Algorithm for the Maximin Objective}\label{sec:maximin}

In the previous section, we showed that it is computationally hard to determine whether a given menu is valid. If even checking whether a single menu is hard, this may suggest that it is hard to design a learning algorithm that implements the optimal menu.

Somewhat counterintuitively, in this section we will show that this is possible, at least in certain settings. In particular, we will design learning algorithms (without any additional constraints) that are nearly optimal for maximizing the learner's \emph{maximin utility}; the worst-case utility they might receive against any type of opponent. 

For any $u_L$ and set of $k$ optimizers $\{u_{O, i}\}_{i=1}^{k}$, let $\OPT = \sup_{\cM} \min_{i \in [k]} V_{L}(\cM, u_{O, i})$ denote the best asymptotic per-round maximin utility a learner can achieve against this set of optimizers (where the supremum is over all valid menus $\cM$). We have the following theorem.

\begin{theorem}
\label{thm:poly_time_general}
For any fixed $\eps > 0$, there exists a learning algorithm $\cA_{\eps}$ (Algorithm~\ref{alg:poly_time_maxmin}) that guarantees a per-round maximin utility of

$$\min_{i \in [k]} V_{L}(\cM(\cA_{\eps}), u_{O, i}) \geq \OPT - \eps$$

\noindent
and runs in time $\poly(n, m, k)$ per iteration. 
\end{theorem}

The primary insight behind the design of the algorithm in Theorem~\ref{thm:poly_time_general} is that even if we do not know whether a specific menu is approachable (in the sense of Blackwell approachability, \ref{thm:blackwell}), we can still attempt to run an algorithmic instantiation of the Blackwell approachability theorem to attempt to approach this set. 

What happens when we do try to approach an unapproachable set $\cM$? Most algorithms for Blackwell approachability do not immediately fail when run on such sets (indeed, if they did, that would contradict the hardness result of \cite{mannor2009approachability}). Instead, they will approach $\cM$ for some time until at some round (depending on the behavior of the adversary) they eventually fail to play an action, certifying that the set $\cM$ is unapproachable. 

We demonstrate this for the specific application of Blackwell Approachability we need in Theorem~\ref{thm:poly_time_general}. Consider the approachability game defined earlier where $u: \Delta_{m} \times \Delta_{n} \rightarrow \Rset^{mn}$ is defined via $u(x, y) = x \otimes y$, and let $\cC(\Phi)$ be the candidate menu for an arbitrary incentive-compatible CSP assignment $\Phi = (\csp_1, \csp_2, \dots, \csp_k) \in \cR$. We will try to force $u$ to approach $\cC(\Phi)$. Note that while we assume $\Phi$ is incentive-compatible, it is not necessarily the case that $\Phi \in \cS$, and so $\cC(\Phi)$ is not necessarily response satisfiable with respect to $u$. 

We say that a learning algorithm is \emph{abortable} if the learning algorithm has the option to terminate its own execution after a finite number of rounds (i.e., outputting an $\abort$ signal instead of playing an action)\footnote{We include a formal definition of what this means (and how this affects the definition of menus) in Appendix~\ref{app:anytime-abortable}.}. The following theorem provides guarantees on an abortable learning algorithm for approaching the set $\cC(\Phi)$ for any incentive-compatible $\Phi$ (in particular, note that the quantity on the LHS of \eqref{eq:blackwell-improper} can be thought of as the maximum extent to which the CSP $\csp = (1/T) \cdot (\sum_{t} x_t \otimes y_t)$ violates the constraints of $\cC(\Phi)$). 

\begin{theorem}\label{thm:blackwell-improper}
For any incentive-compatible CSP assignment $\Phi \in \cR$, there is an efficient abortable learning algorithm $\bwabort(\Phi)$ with the following properties.

\begin{itemize}
    \item If at round $\tau$, $\bwabort(\Phi)$ has not stopped, then its transcript of play up to this round satisfies:

    \begin{equation}\label{eq:blackwell-improper}
    \max_{i \in [k]} \left(\sum_{t=1}^{\tau} (u_{O,i}(x_t, y_t) - u_{O,i}(\csp_i))\right) = O(\sqrt{\tau \log k}).
    \end{equation}

    In particular, the menu $\cM(\bwabort(\Phi))$ is contained within the set $\cC(\Phi)$.
    
    \item If the learning algorithm $\bwabort(\Phi)$ ever aborts, then the set $\cC(\Phi)$ is not response-satisfiable with respect to $u$.
    \item Each 
    $\csp_i \in \Phi$ is contained in the menu $\cM(\bwabort(\Phi))$; i.e., there is a sequence of opponent actions against which this algorithm will never abort and where the CSP will converge to $\csp_i$.
\end{itemize}
\end{theorem}
\begin{proof}
Consider the following learning algorithm:

\begin{itemize}
    \item Initialize an instance of Anytime Hedge \cite{FreundSchapire1996} with $k$ actions (this algorithm has the guarantee that the external regret at round $t$ is at most $O(\sqrt{t \log k})$ for all $t \geq 1$). 
    \item For each round $t$ (for $t \geq 1$), do the following:

    \begin{itemize}
    \item Receive the action $p^{t} \in \Delta_{k}$ from the instance of Anytime Hedge.
    \item Choose an $x_t \in \Delta_{m}$ s.t. 

    $$\sum_{i=1}^{k} p^t_i (u_{O,i}(x_t, y) -  u_{O,i}(\csp_i)) \le 0$$
    
    \noindent
    for all $y \in \Delta_{n}$. We can find such an $x_t$ by solving a linear program; if no such $x_t$ exists, abort.
    \item Play this $x_t$ and observe the opponents action $y_t \in \Delta_{n}$
    \item  Feed Hedge the reward vector $r^t \in [-1, 1]^k$ defined via $r^{t}_i = u_{O,i}(x_t, y_t) - u_{O,i}(\csp_i)$.
    \end{itemize}
\end{itemize}

We claim that this learning algorithm has the desired guarantees. To see this, first note that if the above algorithm ever aborts, by von Neumann's minimax theorem, it must be the case that there exists a $y^{*} \in \Delta_n$ such that

$$\sum_{i=1}^{k} p^t_i (u_{O,i}(x, y^*) - u_{O,i}(\csp_i)) > 0$$

\noindent
for all $x \in \Delta_{m}$. But now, $y^{*}$ is a certificate that $\cC(\Phi)$ is not response-satisfiable with respect to $u$ (i.e., there is no $x$ such that $u(x, y^{*}) = x \otimes y^{*}$ satisfies all the constraints of $\cC(\Phi)$). 

On the other hand, note that if the algorithm has successfully run for $\tau$ rounds, then

\begin{eqnarray*}
& & \max_{i \in [k]} \left(\sum_{t=1}^{\tau} (u_{O,i}(x_t, y_t) - u_{O,i}(\csp_i))\right) \\
&\leq& \max_{i \in [k]} \left(\sum_{t=1}^{\tau} (u_{O,i}(x_t, y_t) - u_{O,i}(\csp_i))\right) - \sum_{t=1}^{\tau} \sum_{i=1}^{k} p^t_i (u_{O,i}(x, y^*) - u_{O,i}(\csp_i)) \\
&=& \max_{i \in [k]} \sum_{t=1}^{\tau} r_{t, i} - \sum_{t=1}^{\tau}  \langle p^t, r_{t}\rangle \\
&=& \Reg(\mathrm{AnytimeHedge}) = O(\sqrt{\tau \log k}).
\end{eqnarray*}

Finally, this algorithm, as written, does not guarantee that $\csp_i$ belongs to its menu. However, we can easily modify this algorithm in the same way as Lemma 3.5 of \cite{arunachaleswaran2024paretooptimal} by adding a cheap talk phase (or implementing one through the details of the protocol) where we cooperate with the opponent to play a sequence of actions converging to the $\csp_i$ of their choice. See also the discussion in Section~\ref{app:explicit_algorithms} for more details.
\end{proof}

\begin{algorithm}
\caption{Polynomial Time Learning Algorithm for Maximin Value}
\label{alg:poly_time_maxmin}
\begin{algorithmic}[1]
\STATE \textbf{Input}: parameter $\eps > 0$
\STATE Initialize $V = 1$ \COMMENT{An upper bound on $\OPT$}
\WHILE{true}
    \STATE Set $\values(V) \gets \{ \phi \mid u_L(\phi) \ge V \}$
    \STATE Set CSP assignment $\Phi(V) \gets \{\phi_i\}_{i \in [k]}$ where:
    \STATE \hspace{1em} $\phi_i \in \argmax_{\phi \in \values(V)} u_{O,i}(\phi)$ 
    \STATE \COMMENT{Note that $\Phi(V)$ is incentive-compatible ($\Phi(V) \in \cR$) since each $\csp_i$ maximizes $u_{O, i}$ over all CSPs in $\cF(V)$.}
    \STATE Run the learning algorithm $\bwabort(\Phi(V))$ (see Theorem~\ref{thm:blackwell-improper}) until it aborts.
    \STATE Set $V \gets V - \eps$ \COMMENT{note that this is only reached if the previous step aborts}
\ENDWHILE
\end{algorithmic}
\end{algorithm}

Based on this idea, we show how to construct an efficient algorithm that maximizes the maximin objective in the following way (described in Algorithm~\ref{alg:poly_time_maxmin}). The core idea is to maintain an upper bound $V$ on $\OPT$. The algorithm runs in epochs, where in each epoch we find the opponent-optimal CSP assignment $\Phi$ supported on CSPs that guarantee the learner at least utility $V$. We then try to approach $\cC(\Phi)$ by running the abortable learning algorithm in Theorem~\ref{thm:blackwell-improper}. If we abort, we decrease $V$ a little bit and start a new epoch. 

We can now prove Thorem~\ref{thm:poly_time_general}.

\begin{proof}[Proof of Theorem~\ref{thm:poly_time_general}]
We first make the following observation: at every step of the game, the variable $V$ in Algorithm~\ref{alg:poly_time_maxmin} never drops below $\OPT - \eps$. To see this, assume $V \leq \OPT$, and consider the CSP assignment $\Phi$ constructed by picking $ \csp_i \in \argmax_{\csp \in \values(V)} u_{O,i}(\csp)$. We will show that $\cC(\Phi)$ is approachable, which implies that the learning algorithm $\cA$ in Theorem~\ref{thm:blackwell-improper} will never abort when run on $\Phi$ (and therefore, $V$ never further decreases after the first time it drops below $\OPT$).

Consider the CSP assignment $\Phi'$ constructed from $\values(\OPT)$. Note that $\cC(\Phi') \subseteq \cC(\Phi)$; this follows almost directly from the fact that $\values(\OPT) \subseteq \values(V)$, implying $u_{O,i}(\csp) \ge u_{O,i}(\csp')$, which in turn guarantees that each constraint hyperplane of $\cC(\csp)$ is weaker than the corresponding constraint in $\cC(\csp')$. 
Now, we also know that an optimal menu $\mathcal{M}$ exists and attains $\OPT$ -- let $\Phi^*$ be the CSP assignment (i.e. $\csp^*_i$ is the profile in $\mathcal{M}$ maximizing the optimizer type's utility). Clearly, $u_{O,i}(\csp^*_i) \le u_{O,i}(\csp'_i) $, which directly implies that $\cC(\csp^*) \subseteq  \cC(\csp')$. But, tautologically, $\mathcal{M} \subseteq \cC(\Phi^*)$ and therefore, since approachability is upwards closed, $\cC(\Phi)$ is approachable.

Armed with this observation, we now examine the menu $\cM(\cA_{\eps})$ of our algorithm. Let $\mathcal{V} = \{1, 1-\eps, 1-2\eps, \dots, V_{\min}\}$ be the set of possible values that any adversary can cause $V$ to take on; note that it contains every real number of the form $1-n\eps$ in the range $(\OPT - \eps, 1]$ (in particular, conversely to the above observation, an adversary can always cause $V$ to decrease if $V > \OPT$ by repeatedly playing a $y \in \Delta_{n}$ that demonstrates that $\cC(V)$ is not approachable). We claim that $\cM(\cA_{\eps})$ has the following properties:

\begin{enumerate}
    \item For each $V \in \mathcal{V}$, every CSP $\csp$ in the CSP assignment $\Phi(V)$ belongs to $\cM(\cA_{\eps})$.
    \item The menu $\cM(\cA_{\eps})$ is a subset of $\cC(\Phi(V_{\min}))$.
\end{enumerate}

The first property is true since each such CSP is explicitly added to one of the menu of one of the potential instances of $\bwabort$. The second property is true\footnote{Here we are implicitly using the fact that the menu of an algorithm constructed in this way is a subset of the convex hull of all of the abortable sub-algorithms. See Lemma~\ref{lem:sub-abortable} in Appendix~\ref{app:anytime-abortable} for a proof of this.} because each of the individual menus $\cM(\bwabort(\Phi(V))$ is contained within $\cC(\Phi(V))$ (by the regret guarantee in Theorem~\ref{thm:blackwell-improper}), and these menus are nested ($\cC(\Phi(1)) \subseteq \cC(\Phi(1-\eps)) \subseteq \dots \subseteq \cC(\Phi(V_{\min}))$). Note that every $\csp_i \in \Phi(V)$ is in fact contained in the set $\cC(\Phi(V))$, i.e., $\Phi(V)$ satisfies the incentive compatibility conditions ($\Phi(V) \in \mathcal{R}$) -- this is because each $\csp_i$ is the maximum value point in $F(V)$ for optimizer payoff $u_{O,i}$. 

We can now evaluate the values $V_{L}(\cM(\cA_{\eps}), u_{O,i})$. Consider the CSP $\csp_i$ that the optimizer $u_{O, i}$ picks from $\cM(\cA_{\eps})$. Let $\csp'_i$ be this optimizer's CSP assignment in $\Phi(V_{\min})$. We claim that $u_{O, i}(\csp'_i) = u_{O, i}(\csp_i)$. Indeed, this follows from the fact that $\cM(\cA_{\eps}) \subseteq \cC(\Phi(V_{\min}))$, and by definition no CSP in $\cC(\Phi(V_{\min}))$ has better utility for optimizer $i$ than their assigned one. On the other hand, we also know that $\csp'_i \in \cM(\cA_{\eps})$. So, since the optimizer breaks ties in favor of the learner, we must have $V_{L}(\cM(\cA_{\eps}), u_{O,i}) \geq u_L(\csp'_i)$. But finally, since $\csp'_i$ was chosen from $\cF(V_{\min})$, it must be the case that $u_{L}(\csp'_i) \geq V_{\min} \geq \OPT - \eps$, as desired.
\end{proof}

It is interesting to consider the strategic situation faced by one of the opponents facing Algorithm~\ref{alg:poly_time_maxmin}. We can think of each stage of Algorithm~\ref{alg:poly_time_maxmin} as a challenge from the learner to the opponent of the form: ``I think I can keep the average CSP in the set $\cC(\Phi(V))$. If I can, it is best for you to play your assigned CSP $\csp_i \in \Phi$ (and I will cooperate with you to reach that CSP). Otherwise, you must demonstrate to me that the set $\cC(\Phi(V))$ is in fact not approachable (at which point I will revise my play).'' In this way, the learner delegates the computational complexity of computing the optimal menu $\cM$ to their opponent, who either must help the learner compute this menu or settle for a potentially worse outcome. As a result, although the guarantee of Theorem~\ref{thm:poly_time_general} is stated only in terms of the guarantee achievable against a perfectly rational (and computationally unconstrained) opponent, it also makes sense (and in fact, is even better) to play Algorithm~\ref{alg:poly_time_maxmin} against computationally-bounded rational opponents. 

Finally, in this section we have restricted ourself to the maximin objective. It is natural to wonder whether we can achieve a similar guarantee for the expected utility objective studied throughout the rest of the paper -- we leave this as a major open question. We note that the current approach takes advantage in several places of the ``one-dimensional'' nature of the maximin objective (for example, it is important that the optimal CSP assignment is simply the one that assigns each opponent their favorite CSP with value at least $V$ for the optimizer), and it is not clear how to adapt it to these more general objectives.

\bibliographystyle{alpha}
\bibliography{references}

\newcommand{\etalchar}[1]{$^{#1}$}
\begin{thebibliography}{CWWZ23}

\bibitem[ACS24]{arunachaleswaran2024paretooptimal}
Eshwar~Ram Arunachaleswaran, Natalie Collina, and Jon Schneider.
\newblock Pareto-optimal algorithms for learning in games.
\newblock In {\em Proceedings of the 25th ACM Conference on Economics and Computation}, EC '24, page 490–510, New York, NY, USA, 2024. Association for Computing Machinery.

\bibitem[AHPY24]{ananthakrishnan2024knowledge}
Nivasini Ananthakrishnan, Nika Haghtalab, Chara Podimata, and Kunhe Yang.
\newblock Is knowledge power? on the (im) possibility of learning from strategic interaction.
\newblock {\em arXiv preprint arXiv:2408.08272}, 2024.

\bibitem[BHP14]{blum2014learning}
Avrim Blum, Nika Haghtalab, and Ariel~D Procaccia.
\newblock Learning optimal commitment to overcome insecurity.
\newblock {\em Advances in Neural Information Processing Systems}, 27, 2014.

\bibitem[Bla56]{blackwell1956analog}
David Blackwell.
\newblock An analog of the minimax theorem for vector payoffs.
\newblock 1956.

\bibitem[BM07]{blum2007external}
Avrim Blum and Yishay Mansour.
\newblock From external to internal regret.
\newblock {\em Journal of Machine Learning Research}, 8(6), 2007.

\bibitem[BMSW18]{braverman2018selling}
Mark Braverman, Jieming Mao, Jon Schneider, and Matt Weinberg.
\newblock Selling to a no-regret buyer.
\newblock In {\em Proceedings of the 2018 ACM Conference on Economics and Computation}, pages 523--538, 2018.

\bibitem[BSV23]{brown2023is}
William Brown, Jon Schneider, and Kiran Vodrahalli.
\newblock Is learning in games good for the learners?
\newblock In {\em Thirty-seventh Conference on Neural Information Processing Systems}, 2023.

\bibitem[CAK23]{collina2023efficient}
Natalie Collina, Eshwar~Ram Arunachaleswaran, and Michael Kearns.
\newblock Efficient stackelberg strategies for finitely repeated games.
\newblock In {\em Proceedings of the 2023 International Conference on Autonomous Agents and Multiagent Systems}, pages 643--651, 2023.

\bibitem[CBL06]{cesa2006prediction}
Nicolo Cesa-Bianchi and G{\'a}bor Lugosi.
\newblock {\em Prediction, learning, and games}.
\newblock Cambridge university press, 2006.

\bibitem[CL23]{chen2023persuading}
Yiling Chen and Tao Lin.
\newblock Persuading a behavioral agent: Approximately best responding and learning.
\newblock {\em arXiv preprint arXiv:2302.03719}, 2023.

\bibitem[CS06]{conitzer2006computing}
Vincent Conitzer and Tuomas Sandholm.
\newblock Computing the optimal strategy to commit to.
\newblock In {\em Proceedings of the 7th ACM conference on Electronic commerce}, pages 82--90, 2006.

\bibitem[CWWZ23]{cai2023selling}
Linda Cai, S~Matthew Weinberg, Evan Wildenhain, and Shirley Zhang.
\newblock Selling to multiple no-regret buyers.
\newblock In {\em International Conference on Web and Internet Economics}, pages 113--129. Springer, 2023.

\bibitem[DDFG23]{dagan2023external}
Yuval Dagan, Constantinos Daskalakis, Maxwell Fishelson, and Noah Golowich.
\newblock From external to swap regret 2.0: An efficient reduction and oblivious adversary for large action spaces.
\newblock {\em arXiv preprint arXiv:2310.19786}, 2023.

\bibitem[DFF{\etalchar{+}}24]{daskalakis2024efficient}
Constantinos Daskalakis, Gabriele Farina, Maxwell Fishelson, Charilaos Pipis, and Jon Schneider.
\newblock Efficient learning and computation of linear correlated equilibrium in general convex games.
\newblock {\em arXiv preprint arXiv:2412.20291}, 2024.

\bibitem[DSS19a]{deng2019prior}
Yuan Deng, Jon Schneider, and Balasubramanian Sivan.
\newblock Prior-free dynamic auctions with low regret buyers.
\newblock {\em Advances in Neural Information Processing Systems}, 32, 2019.

\bibitem[DSS19b]{deng2019strategizing}
Yuan Deng, Jon Schneider, and Balasubramanian Sivan.
\newblock Strategizing against no-regret learners.
\newblock {\em Advances in Neural Information Processing Systems}, 32, 2019.

\bibitem[FS96]{FreundSchapire1996}
Yoav Freund and Robert~E. Schapire.
\newblock Game theory, on-line prediction and boosting.
\newblock In Avrim Blum and Michael~J. Kearns, editors, {\em Proceedings of the Ninth Annual Conference on Computational Learning Theory, {COLT} 1996, Desenzano del Garda, Italy, June 28-July 1, 1996}, pages 325--332. {ACM}, 1996.

\bibitem[FV97]{foster1997calibrated}
Dean~P Foster and Rakesh~V Vohra.
\newblock Calibrated learning and correlated equilibrium.
\newblock {\em Games and Economic Behavior}, 21(1-2):40, 1997.

\bibitem[GKS{\etalchar{+}}24]{guruganesh2024contracting}
Guru Guruganesh, Yoav Kolumbus, Jon Schneider, Inbal Talgam-Cohen, Emmanouil-Vasileios Vlatakis-Gkaragkounis, Joshua~R Wang, and S~Matthew Weinberg.
\newblock Contracting with a learning agent.
\newblock {\em arXiv preprint arXiv:2401.16198}, 2024.

\bibitem[GZG22]{goktas2022robust}
Denizalp Goktas, Jiayi Zhao, and Amy Greenwald.
\newblock Robust no-regret learning in min-max stackelberg games.
\newblock {\em arXiv preprint arXiv:2203.14126}, 2022.

\bibitem[Hic35]{Stack35}
J.~R. Hicks.
\newblock {\em The Economic Journal}, 45(178):334--336, 1935.

\bibitem[HMC00]{hart2000simple}
Sergiu Hart and Andreu Mas-Colell.
\newblock A simple adaptive procedure leading to correlated equilibrium.
\newblock {\em Econometrica}, 68(5):1127--1150, 2000.

\bibitem[HPY24]{haghtalab2024calibrated}
Nika Haghtalab, Chara Podimata, and Kunhe Yang.
\newblock Calibrated stackelberg games: Learning optimal commitments against calibrated agents.
\newblock {\em Advances in Neural Information Processing Systems}, 36, 2024.

\bibitem[KN22a]{kolumbus2022auctions}
Yoav Kolumbus and Noam Nisan.
\newblock Auctions between regret-minimizing agents.
\newblock In {\em Proceedings of the ACM Web Conference 2022}, pages 100--111, 2022.

\bibitem[KN22b]{kolumbus2022and}
Yoav Kolumbus and Noam Nisan.
\newblock How and why to manipulate your own agent: On the incentives of users of learning agents.
\newblock {\em Advances in Neural Information Processing Systems}, 35:28080--28094, 2022.

\bibitem[LCM09]{letchford2009learning}
Joshua Letchford, Vincent Conitzer, and Kamesh Munagala.
\newblock Learning and approximating the optimal strategy to commit to.
\newblock In {\em Algorithmic Game Theory: Second International Symposium, SAGT 2009, Paphos, Cyprus, October 18-20, 2009. Proceedings 2}, pages 250--262. Springer, 2009.

\bibitem[LSV18]{lee2018efficient}
Yin~Tat Lee, Aaron Sidford, and Santosh~S Vempala.
\newblock Efficient convex optimization with membership oracles.
\newblock In {\em Conference On Learning Theory}, pages 1292--1294. PMLR, 2018.

\bibitem[MMSS22]{MMSSbayesian}
Yishay Mansour, Mehryar Mohri, Jon Schneider, and Balasubramanian Sivan.
\newblock Strategizing against learners in bayesian games.
\newblock In {\em Conference on Learning Theory}, pages 5221--5252. PMLR, 2022.

\bibitem[MT09]{mannor2009approachability}
Shie Mannor and John~N Tsitsiklis.
\newblock Approachability in repeated games: Computational aspects and a stackelberg variant.
\newblock {\em Games and Economic Behavior}, 66(1):315--325, 2009.

\bibitem[PPM{\etalchar{+}}08]{paruchuri2008efficient}
Praveen Paruchuri, Jonathan~P Pearce, Janusz Marecki, Milind Tambe, Fernando Ordonez, and Sarit Kraus.
\newblock Efficient algorithms to solve bayesian stackelberg games for security applications.
\newblock 2008.

\bibitem[PR23]{peng2023fast}
Binghui Peng and Aviad Rubinstein.
\newblock Fast swap regret minimization and applications to approximate correlated equilibria.
\newblock {\em arXiv preprint arXiv:2310.19647}, 2023.

\bibitem[SBKK20]{sessa2020learning}
Pier~Giuseppe Sessa, Ilija Bogunovic, Maryam Kamgarpour, and Andreas Krause.
\newblock Learning to play sequential games versus unknown opponents.
\newblock {\em Advances in neural information processing systems}, 33:8971--8981, 2020.

\bibitem[ZT15]{zuo2015optimal}
Song Zuo and Pingzhong Tang.
\newblock Optimal machine strategies to commit to in two-person repeated games.
\newblock In {\em Proceedings of the AAAI Conference on Artificial Intelligence}, volume~29, 2015.

\bibitem[ZZJJ23]{zhao2023online}
Geng Zhao, Banghua Zhu, Jiantao Jiao, and Michael Jordan.
\newblock Online learning in stackelberg games with an omniscient follower.
\newblock In {\em International Conference on Machine Learning}, pages 42304--42316. PMLR, 2023.

\end{thebibliography}

\newpage
\appendix

\section{Omitted Proofs}\label{app:omitted}

\subsection{Proof of Theorem~\ref{thm:nsr_not_optimal}}
Consider a game with $m=3$ learner actions (labeled $A$, $B$ and $C$) and $n=2$ optimizer actions (labeled $R$ and $S$). We will consider a single optimizer type ($k=1$). The utility functions $u_{L}$ and $u_{O}$ are shown in Table~\ref{table:1}.

 \begin{table}
    \centering
    \setlength{\extrarowheight}{2pt}
    \begin{tabular}{*{4}{c|}}
      \multicolumn{2}{c}{} & \multicolumn{2}{c}{Optimizer}\\\cline{3-4}
      \multicolumn{1}{c}{} &  & $R$  & $S$ \\\cline{2-4}
      \multirow{2}*{Learner}  & $A$ & $(0,0)$ & $(3,1)$ \\\cline{2-4}
      & $B$ & $(7 + \epsilon,0)$ & $(2 + \epsilon,1)$ \\\cline{2-4}
     & $C$ & $(7,3)$ & $(1,0)$ \\\cline{2-4}
    \end{tabular} \caption{Payoffs for the counterexample game in the proof of Theorem~\ref{thm:nsr_not_optimal}}
    \label{table:1}
  \end{table}

If the learner plays a no-swap-regret algorithm, the optimizer will respond by playing their Stackelberg equilibrium strategy (by the characterization of $\M_{NSR}$ in Lemma~\ref{lem:nsr_characterization}). The unique Stackelberg equilibrium of this game can be computed to be $(A, S)$, under which the learner earns utility $3$ and the optimizer earns utility $1$.

Now, consider the menu $\M = \conv(\M_{NSR} \cup \{\csp\})$ for the CSP $\csp = \frac{1}{2}(C, R) + \frac{1}{2}(A, S)$. It can be verified that the CSP $\csp$ belongs to $\M_{NR}$ and so $\M$ is a valid no-regret menu. The optimizer must either select a CSP from $\M_{NSR}$ (in which case their best action, as before, is to select their Stackelberg equilibrium) or the CSP $\csp$. If they select $\csp$, the optimizer receives utility $2$ and the learner receives utility $5$. Since this is a higher utility for the optimizer than their Stackelberg leader value, they will select this new CSP, and the learner's resulting utility is thus also strictly better ($3 \rightarrow 5$) under this new menu.

\subsection{Proof of Theorem~\ref{thm:hardness}}

\begin{proof}[Proof of Theorem~\ref{thm:hardness}]
Consider a specific instance of the NP-hard approachability problem in Theorem~\ref{thm:tsitsiklis_hardness}, specified by a vector-valued function $u: \Delta_m \times \Delta_n \rightarrow \Rset^{d}$. We will show how to (in polynomial time) construct an $m$-by-$n$ game $G$ with $k=d$ optimizer payoffs $u_{O, i}$ and a CSP assignment $\Phi$ such that $\cC(\Phi)$ is a valid menu iff the nonnegative orthant is approachable in the vector-valued game specified by $u$.

We will construct our game as follows. We will set $k = d$ and $u_{O, i}(x, y) = -u_{i}(x, y)$ for each $i \in [d]$. We will set $u_L$ arbitrarily (note that it does not feature in the definition of $\cC(\Phi)$ or whether it is a valid menu). Finally, we will construct our CSP assignment $\Phi = (\csp_1, \csp_2, \dots, \csp_k)$ as follows. For each $i \in [k]$ we search over all the $mn$ pure action pairs, finding a pair $a^{-}_i, b^{-}_i$ where $u_{i}(a^{-}_i, b^{-}_i) \leq 0$, and a pair $a^{+}_i, b^{+}_i$ where $u_{i}(a^{+}_i, b^{+}_i) \geq 0$. Note that if we cannot find the first pair, then the $i$th coordinate of $u(x, y)$ is always non-negative and we can drop it from $u$; on the other hand if we cannot find the second pair, then the $i$th coordinate of $u(x, y)$ is never non-negative and we can simply return that the non-negative orthant is not approachable. Finally, set $\csp_i$ to the convex combination of these two pure strategy profiles such that $u_{O,i}(\csp_i) = -u_{i}(\csp_i) = 0$.

Now, recall (by Lemma \ref{lem:csp-expand-util})  that $\cC(\Phi)$ is a valid menu iff the set $\cU(\Phi)$ is approachable under the vector-valued payoff $v(x, y) = (u_{O, 1}(x, y), u_{O, 2}(x, y), \dots, u_{O, k}(x, y))$. But for this $\Phi$ (since $u_{O, i}(\csp_i) = 0$), $\cU(\Phi)$ is exactly the negative orthant, and therefore this problem is equivalent to the problem of testing whether the non-negative orthant is approachable under our original payoff $u(x, y) = -v(x,y)$. This completes the reduction and the proof.
\end{proof}

\section{From Learning Algorithms to Menus}
\label{app:algtomenus}

In this appendix we outline the details of how to translate a learning algorithm $\A$ (described as a family of functions from histories to next actions) to a menu $\M$ (described as a convex subset of $mn$-dimensional space); in particular, we show that the asymptotic utilities received by the learner in a game where they commit to a specific algorithm correspond to the asymptotic utilities they receive in the menu commitment form of this problem. This follows from results contained in \cite{arunachaleswaran2024paretooptimal}, and we largely follow their presentation.


\subsection{Asymptotic Algorithmic Commitment}

We begin by establishing the corresponding utility-theoretic quantities for the problem of algorithmic commitment. Recall that we assume that the learner is able to see $u_{L}$ and $\optdist$ before committing to their algorithm $\cA$, but not the  $u_{O}$ privately drawn by the optimizer. The optimizer, who knows $u_L$ and $u_{O}$, will approximately best-respond by selecting a sequence of actions that approximately (up to sublinear $o(T)$ factors) maximizes their payoff. They break ties in the learner's favor. Formally, for each $T$, let

\begin{equation*}
    V_{O}(\cA, u_{O}, T) = \sup_{(x_1, \dots, x_T) \in \Delta_{m}^{T}} \frac{1}{T}\sum_{t=1}^{T} u_{O}(x_t, y_t)
\end{equation*}

\noindent
represent the maximum per-round utility of the optimizer with payoff $u_O$ playing against $\cA$ in a $T$ round game (here and throughout, each $y_t$ is determined by running $A^{T}_{t}$ on the prefix $x_1$ through $x_{t-1}$). For any $\eps > 0$, let 

\begin{equation*}
\cX(\cA, u_{O}, T, \eps) = \left\{(x_1, x_2, \dots, x_T) \in \Delta_{m}^T \bigm| \frac{1}{T}\sum_{t=1}^{T} u_{O}(x_t, y_t) \geq V_{O}(\cA, u_O, T) - \eps \right\}
\end{equation*}

\noindent
be the set of $\eps$-approximate best-responses for the optimizer to the algorithm $\cA$. Finally, let

\begin{equation*}
V_{L}(\cA, u_O, T, \eps) = \sup_{(x_1, \dots, x_T) \in \cX(\cA, u_O, T, \eps)} \frac{1}{T}\sum_{t=1}^{T} u_{L}(x_t, y_t)
\end{equation*}

\noindent
represent the maximum per-round utility of the learner under any of these approximate best-responses. 

We are concerned with the asymptotic per-round payoff of the learner as $T \rightarrow \infty$ and $\eps \rightarrow 0$. Specifically, let

\begin{equation}\label{eq:learner_value}
V_{L}(\cA, u_O) = \lim_{\eps \rightarrow 0} \liminf_{T \rightarrow \infty} V_{L}(\cA, u_O, T, \eps).
\end{equation}

\noindent
Note that the outer limit in \eqref{eq:learner_value} is well-defined since for each $T$, $V_{L}(A, u_O, T, \eps)$ is decreasing in $\eps$ (being a supremum over a smaller set). 

We also define an approximate measure $V^\varepsilon_L$ in the following manner:

\begin{equation}\label{eq:tiebreak_learner_value}
V^\varepsilon_{L}(\cA, u_O) =  \liminf_{T \rightarrow \infty} V_{L}(\cA, u_O, T, \eps).
\end{equation}

Once again, this is well defined since all payoffs are bounded in $[-1,1]$ and the lim inf always exists. This approximate measure should be understood as the optimizer best-responding approximately up to a $\varepsilon$ additive error and tie-breaking in favor of the learner.

The final objective of the learner is the expectation of the learner's asymptotic per-round payoff over the draws of the optimizer's payoff function:

\begin{equation}\label{eq:learner_value2}
V_{L}(\cA, \optdist) = \mathbb{E}_{u_O \sim \optdist} [ V_{L}(\cA, u_O) ] = \sum_{i=1}^k \alpha_i V_{L}(\cA, u_{O,i}) .
\end{equation}

\noindent
We define $V_L^\varepsilon(\cA,\optdist)$ similarly.

\subsection{Menus as Limit Objects}

We now switch our attention to formally defining menus. For an algorithm $A^{T}$ that defines the algorithm $\A$ run for $t$ rounds, define the \emph{menu} $\M(A^T) \subseteq \Delta_{mn}$ of $A^{T}$ to be the convex hull of all CSPs of the form $\frac{1}{T}\sum_{t=1}^T x_t \otimes y_t$, where $(y_1, y_2, \dots, y_T)$ is any sequence of optimizer actions and $(x_1, x_2, \dots, x_T)$ is the response of the learner to this sequence under $A^T$ (i.e., $x_t = A^{T}_t(y_1, y_2, \dots, y_{t-1})$).

If a learning algorithm $\cA$ has the property that the sequence $\M(A^1), \M(A^2), \dots$ converges under the Hausdorff metric\footnote{The \emph{Hausdorff distance} between two bounded subsets $X$ and $Y$ of Euclidean space is given by $d_H(X, Y) = \max(\sup_{x \in X} d(x, Y), \sup_{y \in Y} d(y, X))$, where $d(a, B)$ is the minimum Euclidean distance between point $a$ and the set $B$.}, we say that the algorithm $\cA$ is \emph{consistent} and call this limit value the \emph{asymptotic menu} $\M(\cA)$ of $\cA$. More generally, we will say that a subset $\M \subseteq \Delta_{mn}$ is an asymptotic menu if it is the asymptotic menu of some consistent algorithm. It is possible to construct learning algorithms that are not consistent (for example, imagine an algorithm that runs multiplicative weights when $T$ is even, and always plays action $1$ when $T$ is odd); however even in this case we can find subsequences of time horizons where this converges and define a reasonable notion of asymptotic menu for such algorithms. See the appendix of~\cite{arunachaleswaran2024paretooptimal} for a discussion of this phenomenon.

For any asymptotic menu $\M$, define $V_{L}(\M, u_O)$ to be the utility the learner ends up with under this process. Specifically, define $V_{L}(\M, u_O, \varepsilon) = \max \{ u_{L}(\csp) \mid u_O(\csp) \ge \max_{\csp \in \M} u_{O}(\csp) - \varepsilon\}$ and $V_{L}(\M, u_O) = V_{L}(\M, u_O, 0)$. We can verify that this definition is compatible with our previous definition of $V_{L}$ as a function of the learning algorithm $\cA$ (see~\cite{arunachaleswaran2024paretooptimal} for proof). 

\begin{lemma}\label{lem:menu_alg_equiv}
For any learning algorithm $\cA$, $V_{L}(\M(\cA), u_O) = V_{L}(\cA, u_O)$ and $V_{L}(\M(\cA), u_O , \varepsilon) = V_{L}^\varepsilon(\cA, u_O)$.
\end{lemma}

In particular, Lemma~\ref{lem:menu_alg_equiv} establishes that it is equivalent to optimize over the set of valid menus $\M$ (menus that correspond to $\M(\cA)$ for some learning algorithm $\cA$), and that any optimal solution to the menu commitment problem corresponds to an optimal solution to the algorithmic commitment problem (and vice versa).

\subsection{From Optimal Menus to Optimal Algorithms}
\label{app:explicit_algorithms}
Most of our constructions throughout this paper will return solutions in the form of optimal menus to commit to. However, ultimately, our goal is to run explicit learning algorithms that actually implement these menus. 

In general, given a menu $\cM$ and an efficient separation oracle for this menu, \cite{arunachaleswaran2024paretooptimal} give a general procedure (via an algorithmic implementation of Blackwell approachability) for constructing an algorithm $\cA$ with $\cM(\cA) = \cM$. For completeness, in this section we describe at a high-level what this procedure looks like for the constructions we introduce in this paper (see Lemmas 3.4 and 3.5 of \cite{arunachaleswaran2024paretooptimal} for details).

\paragraph{Optimal No-Regret Commitment} The optimal menu returned by our algorithm for optimal no-regret commitment always has the form $\M = \conv(\M_{NSR} \cup \{\csp_1, \csp_2, \dots, \csp_k\})$, where $\Phi = (\csp_1, \csp_2, \dots, \csp_k)$ is the optimal CSP assignment we computed. We can implement this menu as follows: first, we let the opponent pick any of the CSPs $\csp_i$. We then provide the opponent with a schedule of action profiles (specifying the action pair played at time $t$ for all $1 \leq t \leq T$) whose CSP converges to $\csp_i$ while never leaving $\M$. Finally, if the opponent ever defects from this schedule, we fall back to playing any no-swap-regret algorithm $\cA_{NSR}$. 

\paragraph{Optimal General Commitment} The optimal menu returned by our algorithm for optimal general commitment always has the form $\M = \cC(\Phi, \eps)$, for some CSP assignment $\Phi = (\csp_1, \csp_2, \dots, \csp_k)$ and $\eps > 0$. Note that since $\cC(\Phi, \eps)$ is defined by at most $k$ half-space constraints, it is easy to construct a separation oracle for $\M$. We now proceed essentially identically to no-regret commitment, with the exception that instead of playing a no-regret algorithm $\cA_{NSR}$ as our ``fallback'' algorithm, we play a Blackwell approachability algorithm $\cA'$ that guarantees that the CSP converges to the set $\M$ (we can implement this efficiently with the above-mentioned separation oracle). Technically, as written, this procedure may converge to a sub-menu $\M' \subseteq \M$  with the property that every $\csp_i$ is contained in $\M'$, but it is easy to check that such a menu is also optimal (alternatively, we can do the padding described in Lemma 3.5 of \cite{arunachaleswaran2024paretooptimal} to design an algorithm $\cA$ where $\cM(\cA)$ is the full menu $\M$).

\subsection{Anytime and Abortable Learning Algorithms}\label{app:anytime-abortable}

In Section~\ref{sec:maximin}, our learning algorithm makes use of a sub-algorithm which is both \emph{anytime} (i.e., can be run without knowing the time horizon $T$ in advance) and \emph{abortable} (i.e., may decide to stop after a given round). In this Appendix, we define both of these classes of learning algorithms (and associated concepts of menus) formally.

An \emph{anytime learning algorithm} $\cA$ is a learning algorithm where all the horizon-dependent algorithms agree on every prefix (i.e., the functions $A^{T}_{t}: \Delta_{n}^{t-1} \rightarrow \Delta_{m}$ are the same for every $T \geq t$). We can therefore think of an anytime learning algorithm $\cA$ as being specified by a single infinite collection of functions $A_1, A_2, \dots, A_{t},\dots$, where $A_{t}: \Delta_{n}^{t-1} \rightarrow \Delta_{m}$ maps the transcript of play $y_1, y_2, \dots, y_{t-1}$ by the optimizer to the next action $x_t$ of the learner. Since anytime learning algorithms are themselves learning algorithms, we can define their menus the same way as before.

An \emph{abortable learning algorithm} is a variant of an anytime\footnote{It also makes sense to define finite-horizon abortable learning algorithms, but the only abortable learning algorithm we construct in this paper (Lemma~\ref{thm:blackwell-improper}) is also an anytime learning algorithm.} learning algorithm where each round function $A_t$ has the option of ``aborting''; formally, such an algorithm is a collection of round functions $A_t$ (for all $t \geq 1$), where each $A_t$ is now a function $A_t: \Delta_{n}^{t-1} \rightarrow \Delta_{m} \cup \{\abort\}$, and it is understood that a transcript of play stops the first time the learner outputs $\abort$. This also requires a slight change in the definition of menu: we now define the finite time-horizon menu $\cM(\cA^{T}) \subseteq \Delta_{mn}$ to be the convex hull of all CSPs of all \emph{non-aborting} length $T$ transcripts where the learner plays the sequence of actions $x_1, x_2, \dots, x_T \in \Delta_m$ against the optimizer's sequence of actions $y_1, y_2, \dots, y_T \in \Delta_n$ (notably, $x_t \neq \abort$ for all $t \in [T]$; if there is no non-aborting sequence transcript of length $T$, the menu is the empty set). We similarly define $\cM(\cA)$ to be the limit of the convex sets $\cM(\cA^{T})$ in the Hausdorff metric as $T \rightarrow \infty$. 

Our main use of abortable learning algorithms will be to construct a specific (non-abortable) learning algorithm ($\bwabort$, in Theorem~\ref{thm:blackwell-improper}), implemented by running a specific sequence of abortable learning algorithms in order, switching to the next one as soon as the current one aborts. The following lemma demonstrates how we can bound the menu of the overall algorithm by the menu of sub-algorithms.

\begin{lemma}\label{lem:sub-abortable}
Let $\cB_{1}, \cB_{2}, \dots, \cB_{K}$ be a sequence of $K$ abortable learning algorithms with the guarantee that $\cB_{K}$ never aborts. Consider the (standard) learning algorithm $\cA$ constructed by running, in order, each of the algorithms $\cB_{k}$ until it aborts, and then switching to the next algorithm $\cB_{k+1}$. Then $\cM(\cA) \subseteq \conv(\cM(\cB_1) \cup \cM(\cB_2) \cup \dots \cup \cM(\cB_K))$.
\end{lemma}
\begin{proof}
Consider an arbitrary length $T$ transcript of $\cA$, with corresponding CSP $\csp \in \cM(\cA^T)$. Let $T_k$ be the number of rounds in this transcript where the algorithm $\cB_{k}$ was run, and let $\csp_k \in \cM(\cB_{K}^{T_k})$ be the CSP corresponding to this subtranscript. We therefore have that

$$T\csp = \sum_{k=1}^{K} T_k \csp_k$$

Let $S_{\mathrm{long}}$ be the subset of indices $k$ where $T_k \geq T^{0.9}$, and let $S_{\mathrm{short}} = [k] \setminus S_{\mathrm{long}}$ contain the remaining indices. We can partition the above sum accordingly:

$$T\csp = \sum_{k \in S_{\mathrm{long}}} T_k \csp_k + \sum_{k' \in S_{\mathrm{short}}} T_{k'} \csp_{k'}.$$

\noindent
Now, note that $\sum_{k' \in S_{\mathrm{short}}} T_{k'} \leq KT^{0.9}$, so

$$\frac{1}{T}\left|\sum_{k' \in S_{\mathrm{short}}} T_{k'}\csp_{k'}\right| \leq KT^{-0.1}$$

\noindent
(where here the norm is the $\ell_2$ norm). In particular, letting $T_{\mathrm{long}} = \sum_{k \in S_{\mathrm{long}}} T_k$ (and noting that $|1/T_{\mathrm{long}} - 1/T| \leq T_{\mathrm{short}}/(TT_{\mathrm{long}}) \leq KT^{-0.1}/T_{\mathrm{long}}$), this implies that  

$$\left|\csp - \frac{1}{T_{\mathrm{long}}} \sum_{k \in S_{\mathrm{long}}} T_k \csp_k\right| \leq 2KT^{-0.1}.$$

The CSP $\csp$ is therefore within distance $o(T)$ of the CSP $\csp' = \sum_{k \in S_{\mathrm{long}}} (T_k/T_{\mathrm{long}}) \csp_k$. This itself is a convex combination of CSPs $\csp_k \in \cM(\cB_{K}^T)$ and therefore has the property that

$$\csp' \in \conv\left(\bigcup_{k \in S_{\mathrm{long}}} \cM(\cB_{k}^{T_k})\right),$$

\noindent
where each $T_k \geq T^{0.9}$. As we take $T \rightarrow \infty$, the above set must therefore converge in Hausdorff distance to the convex hull of the sets $\cM(\cB_k)$ for $k \in S_{\mathrm{long}}$. It follows that for sufficiently large $T$, any CSP $\csp \in \cM(A^{T})$ must be arbitrarily close to some CSP in $\conv(\cM(\cB_1) \cup \cM(\cB_2) \cup \dots \cup \cM(\cB_K))$, establishing the theorem.
\end{proof}

\end{document}